%% file: main.tex
\documentclass[sigplan,nonacm,balance=false]{acmart}

\usepackage{listings}
\usepackage{listings-rust}
\ifdefined\AnonMode
\fi

\ifdefined\AnonMode
  \newcommand{\sysname}{Kare\xspace}
  \newcommand{\sysnamens}{Kare}
  \newcommand{\enginename}{Patter\xspace}
  \newcommand{\enginenamens}{Patter}
  \newcommand{\cratename}{kare}
  \newcommand{\papertitle}{Ownership-Based Safety for GPU Tile Programming}
  \newcommand{\ifanon}[2]{#1}
  \newcommand{\syscite}{}
  \newcommand{\enginecite}{}
\else
  \newcommand{\sysname}{cuTile Rust\xspace}
  \newcommand{\sysnamens}{cuTile Rust}  
  \newcommand{\enginename}{Grout\xspace}
  \newcommand{\enginenamens}{Grout}
  \newcommand{\cratename}{cutile}
  \newcommand{\papertitle}{Fearless Concurrency on the GPU}
  \newcommand{\ifanon}[2]{#2}
  \newcommand{\syscite}{~\cite{cutile_rust}}
  \newcommand{\enginecite}{~\cite{grout}}
\fi

\newif\ifappRuntime \appRuntimefalse   
\ifappRuntime \newcommand{\AppRuntime}[1]{#1} \else \newcommand{\AppRuntime}[1]{} \fi
\newif\ifappGraph \appGraphfalse       
\ifappGraph \newcommand{\AppGraph}[1]{#1} \else \newcommand{\AppGraph}[1]{} \fi
\newif\ifappOverlap \appOverlapfalse   
\ifappOverlap \newcommand{\AppOverlap}[1]{#1} \newcommand{\appDTitle}{Supplementary Work-Distribution Data}
\else \newcommand{\AppOverlap}[1]{} \newcommand{\appDTitle}{Completion-Path Costs} \fi
\newif\ifappCensusGemm \appCensusGemmfalse 
\ifappCensusGemm \newcommand{\AppCensusGemm}[1]{#1} 
\else \newcommand{\AppCensusGemm}[1]{}  \fi
\newif\ifappRaceSource \appRaceSourcefalse 
\ifappRaceSource \newcommand{\AppRaceSource}[1]{#1} \else \newcommand{\AppRaceSource}[1]{} \fi

\usepackage{tikz}
\usetikzlibrary{arrows.meta,positioning,fit,backgrounds,calc,decorations.pathreplacing,shapes.geometric}
\usepackage{pgfplots}
\pgfplotsset{compat=1.18}
\usepgfplotslibrary{groupplots,fillbetween}
\newif\ifexpRatioPlots \expRatioPlotstrue
\pgfplotsset{
  every axis/.append style={
    label style={font=\footnotesize},
    tick label style={font=\footnotesize},
    title style={font=\footnotesize},
  },
  every axis legend/.append style={font=\footnotesize},
}
\usepackage{booktabs}
\usepackage{colortbl}
\usepackage{subcaption}
\usepackage{xspace}
\usepackage{enumitem}

\AtBeginDocument{\raggedbottom}

\definecolor{lstKeyword}{HTML}{1A3E7A}
\definecolor{lstTrait}{HTML}{A31515}
\definecolor{lstType}{HTML}{00695C}
\definecolor{lstMacro}{HTML}{6F42C1}
\definecolor{lstString}{HTML}{8A4B08}
\definecolor{lstComment}{HTML}{6A737D}
\definecolor{exPython}{HTML}{E69F00}
\definecolor{exUnsafe}{HTML}{009E73}
\definecolor{exSafe}{HTML}{0072B2}
\definecolor{exPeak}{HTML}{888888}
\definecolor{expSyncCh}{HTML}{E69F00}
\definecolor{expAsync}{HTML}{009E73}
\definecolor{expGraph}{HTML}{0072B2}
\definecolor{expGraphCap}{HTML}{5BA3D0}
\definecolor{ownMut}{HTML}{F4C7C3}
\definecolor{ownShared}{HTML}{D9E8D4}
\definecolor{ownP1}{HTML}{F9DEC9}
\definecolor{ownP2}{HTML}{FBEFC4}
\definecolor{ownP3}{HTML}{E4DCF0}
\begin{document}

\title{\papertitle}

\author{Melih Elibol}
\email{melibol@nvidia.com}
\affiliation{%
  \institution{NVIDIA}
  \country{USA}
}

\author{Jared Roesch}
\email{jroesch@nvidia.com}
\affiliation{%
  \institution{NVIDIA}
  \country{USA}
}

\author{Isaac Gelado}
\email{igelado@nvidia.com}
\affiliation{%
  \institution{NVIDIA}
  \country{USA}
}

\author{Eric Buehler}
\email{eric@huggingface.co}
\affiliation{%
  \institution{Hugging Face}
  \country{USA}
}

\author{Michael Garland}
\email{mgarland@nvidia.com}
\affiliation{%
  \institution{NVIDIA}
  \country{USA}
}

\begin{abstract}
\input{sections/abstract}
\end{abstract}

\maketitle

\input{sections/introduction}
\input{sections/overview}
\input{sections/design}
\input{sections/placement}
\input{sections/execution}
\input{sections/evaluation}
\input{sections/limitations}
\input{sections/related_work}
\input{sections/conclusion}

\ifdefined\AnonMode\else
\renewenvironment{acks}{%
  \section*{\acksname}%
  \phantomsection\addcontentsline{toc}{section}{\acksname}%
}{}
\fi
\input{sections/acknowledgements}
\ifdefined\AnonMode
\paragraph*{Generative AI disclosure}
We used LLM-based coding agents to build the benchmark harnesses, plotting scripts, and measurement runs; the authors directed every action, reviewed the generated code, and checked every result.
\fi

\bibliography{references}

\input{sections/appendix}

\end{document}

%% file: sections/abstract.tex
Rust has made safe systems programming practical on the CPU, but
writing custom GPU kernels in Rust still forces programmers outside
the language's ownership guarantees. We present \sysname, a tile-based system for safe, idiomatic GPU
kernel authoring in Rust that compiles kernels to Tile~IR.
\sysname extends Rust's ownership discipline to tile-based GPU
kernels: mutable outputs are split into disjoint pieces, kernel
launches preserve the host-side ownership contract, and the Rust
compiler enforces the same ownership rules inside the kernel. We prove
the safe surface data-race-free under Tile~IR's memory model.
However, bounds safety still requires runtime checks. The compiler therefore eliminates
checks it can prove redundant and, where possible, moves others out of
the kernel into host-side launch preconditions. On the host, the same
ownership contract carries through a composable execution model that
runs the same operations synchronously, under async/await, or as CUDA
graph replay, with async at parity with synchronous execution.

Our evaluation shows that these abstractions preserve performance
on high-end GPUs. On the NVIDIA B200 GPU, \sysname achieves
7~TB/s for element-wise operations and 2.1~PFlop/s for GEMM
(98\% of cuBLAS), on par with cuTile Python.
\enginename, a Qwen3 inference engine built on \sysname, reaches 171
generated tokens/s for Qwen3-4B on the NVIDIA GeForce RTX~5090 and 82
for Qwen3-32B on the B200 in batch-1 decode, at parity with vLLM and
SGLang and consistent with a memory bandwidth roofline sanity check.

%% file: sections/introduction.tex
\section{Introduction}
\label{sec:introduction}

The Rust language makes programming with static safety guarantees both
practical and efficient.  These qualities are particularly attractive in
applications that must coordinate the activity of many threads working
on shared data objects, and modern AI frameworks and LLM inference
engines are important examples~\cite{candle,burn,mistralrs}.  Such
applications perform computations where parallelism is abundant, and
significant portions of their workloads are built around tensor
operations. These tensor workloads are well suited to GPUs.
However, GPUs do not fit comfortably into the
current Rust ecosystem.

GPU-accelerated applications contain code for both the CPU (host) and
GPU (device) processors.  The host code allocates data, builds kernel
arguments, and submits work for execution on the device. Device kernels
execute in a single-program multiple-data (SPMD) fashion: thousands of
threads start at the same entry point, scheduled onto the GPU's
processors in thread blocks whose threads share on-chip memory and
synchronize explicitly.  Kernel arguments cross this boundary as raw pointers and scalars, so
whatever a safe host language such as Rust established about ownership
stops at the launch site, and correct synchronization and the prevention
of data races inside the kernel are left in the hands of the programmer.

The Rust compiler (\texttt{rustc}) is able to generate device code for
NVIDIA GPUs via the LLVM PTX backend, and other projects have
experimented with other code generation
paths \cite{rust_cuda,rust_gpu,cuda_oxide}.  However, in each of these
cases, kernels still require \lstinline|unsafe| for basic performance
techniques, such as the use of shared memory.
As GPU kernels are increasingly generated by AI tools rather than
written by hand, this gap grows more consequential: The burden shifts
from writing kernels to trusting them.

With \sysnamens\syscite, we extend Rust's safety guarantees across both
host and device code.
Many high-performance CUDA kernels partition their outputs among thread
blocks and share read-only inputs.
We therefore adopt a tile-based model for kernels~\cite{cutile_python, cuda_tile},
which exposes this pattern at a granularity where Rust's ownership
rules apply directly.
Instead of launching a grid of thread blocks, a tile kernel is launched
as a grid of single-threaded tile programs, each operating on immutable,
fixed-size tiles loaded from and stored to tensors in memory. The
compilation stack realizes each tile program as a thread block and
manages the shared memory and synchronization within it.
Kernels compose loads, stores, and whole-tile operations whose memory
footprints are visible in their types. Rust's aliasing discipline
therefore maps directly onto the units the hardware schedules.

Existing tile DSLs, however, leave the destination of each store to the
program. In the head-permutation example of \S\ref{sec:design:race},
the store's destination index omits the head coordinate. The kernel
compiles and runs, but tile programs processing different heads write
to the same destination, causing a data race.
We define tensor types that allow immutable tensors to be passed from
host to device as is, while mutable tensors are partitioned before
launch and mapped to mutable sub-tensors within the kernel grid. This
ensures that no two tile programs can write to the same sub-tensor.

Bounds safety also requires checks. Like Rust, \sysname checks every
access, but a GPU kernel is on the hot path, so our compiler resolves each
check obligation where it is cheapest: At specialization time, at the
launch site on the host, at a loop preheader in the kernel, or, when
necessary, at the access site. In the GEMM of
\S\ref{sec:placement}, every check resolves before the kernel runs, and in
attention the remaining checks move to a loop preheader, so neither
kernel has a check in its hot loop.

The tensor operations performed in host code compose into lazy GPU work
that can be executed synchronously, submitted asynchronously to a
stream, or captured as a CUDA Graph for later replay. The result is a
GPU programming model that is safe end to end, from host allocation
through kernel execution and back, with performance on par with
state-of-the-art GPU libraries and inference engines
(\S\ref{sec:evaluation}).

Our main programming-language contribution is an ownership-checked tile
programming model with bounds-check placement. Our systems contribution
is a host execution model that combines bounded spin polling with a
completion reactor, making asynchronous GPU execution as fast as
synchronous execution at microsecond timescales. We make the following
contributions:
\begin{enumerate}[topsep=5pt,itemsep=4pt,parsep=1pt,leftmargin=1.4em]
  \item A safe, high-performance programming model for Rust: host
    tensors become device-side tensor and partition views, tile kernels
    execute with single-threaded semantics, and a data-race-freedom
    theorem holds for safe kernels and their launches under Tile~IR's
    weakly ordered memory model (Theorem~\ref{thm:drf}, Appendix~\ref{app:drf}). The partition map keeps each program's
    mutable sub-tensors disjoint, and branded partition indices let a
    program iterate over them with every store proven in bounds at no
    runtime cost. Unsafe low-level memory operations remain available
    for building new safe abstractions.
  \item A bounds-check placement mechanism that migrates checks across
    the CPU/GPU boundary. Unoptimized checks add 17--22\% to runtime on
    the evaluated GEMM and attention kernels; optimized placement brings
    performance to parity with their unchecked versions
    (\S\ref{sec:eval:checks}).
  \item A composable host execution model whose generated launch code
    prevents host access while GPU work is in flight and returns
    ownership in the same form after completion, with async at parity
    with synchronous execution at microsecond scale.
\end{enumerate}

Our benchmarks on an NVIDIA DGX B200 reach roughly 7~TB/s for
element-wise operations and 2.1~PFlop/s for GEMM, about 98\% of cuBLAS
and on par with cuTile Python.  End-to-end, \enginename,
a Qwen3 LLM inference engine built upon \sysname, runs at parity with
vLLM and SGLang, with decode throughput consistent with the
memory bandwidth roofline.
\ifanon{}{Beyond our own case study, \sysname is in use by two Rust
inference stacks: Candle~\cite{candle}, Hugging Face's machine-learning
framework, exposes it through an opt-in feature for writing JIT-compiled
CUDA kernels as custom operations, shipping a fused mixture-of-experts
grouped matrix multiplication written in it; and
mistral.rs~\cite{mistralrs} depends on it for optional kernel
acceleration.}

%% file: sections/overview.tex
\section{Overview}
\label{sec:overview}

This section walks through a complete \sysname program,
Listing~\ref{lst:vec_add}: host code constructs tensors and partitions
the mutable output, and device code runs a tile program for each output
partition. Sections~\ref{sec:design}--\ref{sec:execution} present the
design. \S\ref{sec:design} gives the safety model, which determines what
the compiler proves statically and which bounds checks remain at run
time. \S\ref{sec:placement} describes where the compiler resolves those
checks, at specialization time, on the host as launch preconditions, or
in the kernel. \S\ref{sec:execution} describes the host execution model
that performs the launch and returns ownership when the work completes.

\begin{figure}[tb]
\begin{lstlisting}[caption={Element-wise addition in \sysname. The host
  (lines~16--22) partitions the output and passes inputs as shared reads;
  the kernel (lines~2--15) operates on tiles with
  single-threaded semantics. The program contains no \texttt{unsafe}
  keywords.},
  label={lst:vec_add}]
use cutile::prelude::*;
#[cutile::module]
mod kernel {
  use cutile::core::*;
  #[cutile::entry()]
  fn add<const B: i32>(
    z: &mut Tensor<f32, {[B]}>,   // exclusive write
    x: &Tensor<f32, {[-1]}>,      // shared read
    y: &Tensor<f32, {[-1]}>,      // shared read
  ) {
    let tx = x.load_like(z);
    let ty = y.load_like(z);
    z.store(tx + ty);
  }
}
fn main() -> Result<()> {
  let x = api::ones::<f32>([1024]);
  let y = api::ones::<f32>([1024]);
  let z = api::zeros::<f32>([1024]).partition([128]);
  let (_z, _x, _y) = kernel::add(z, x, y).sync()?;
  Ok(())
}
\end{lstlisting}
\end{figure}

On the host (lines~16--22), the program creates two input tensors,
partitions the mutable output into 128-element chunks, and calls
\lstinline|kernel::add(z, x, y)|, which returns a value that
represents the device operation and holds its arguments. Nothing runs until an
execution mode is chosen: \lstinline|.sync()| launches the operation,
waits for the stream, and returns the same tensors. The same value can
instead be awaited as a future or captured into a CUDA graph
(\S\ref{sec:execution}).

The kernel signature (lines~6--10) carries Rust's ownership rules into
the device code. The output parameter is an exclusive mutable sub-tensor
(\lstinline|&mut Tensor|), while the inputs are shared read-only tensors
(\lstinline|&Tensor|). The body (lines~11--13) loads a tile from each
input, adds the tiles, and stores the result through the exclusive output
reference. Here, \lstinline|x.load_like(z)| partitions \lstinline|x| like
\lstinline|z| and loads the tile at the partition index this program owns
in \lstinline|z|.

Two compilers process this program (Figure~\ref{fig:compilation_pipeline}).
\texttt{rustc} type-checks the host code and the macro-generated Rust
version of the kernel. The module macro also embeds the kernel's original
AST in the host binary. At first
launch, \sysname's just-in-time (JIT) compiler lowers that AST to Tile~IR
and enforces the same access rules on the device, inserting runtime
checks where bounds cannot be proved statically
(\S\ref{sec:placement}). The implementation is
about 90 thousand lines of Rust, tests included: the proc macro (7.5k),
the JIT compiler and Tile~IR layer (50k), the tensor runtime (13k), and
the CUDA driver and async layers (20k).

\begin{figure}[t]
  \centering
  \input{figures/src/compilation_pipeline}
  \caption{A kernel from source to GPU: what the module macro emits at
    Rust compile time (top) and what the runtime JIT compiles at first
    launch (bottom, dashed arrow).}
  \label{fig:compilation_pipeline}
\end{figure}

Rust's relevant rule is aliasing XOR mutability: A value may have
either one mutable reference (\lstinline|&mut T|) or any number of
immutable references (\lstinline|&T|), but not both simultaneously.
Rust provides safe abstractions where the compiler can verify that
invariant, runtime checks when it is not known statically, and explicit
\lstinline|unsafe| code where the programmer supplies it. \sysname
follows the same structure (\S\ref{sec:design}) and places the runtime
checks where they are cheapest (\S\ref{sec:placement}).

\subsection{GPU Execution Model}
\label{sec:bg:gpu-programming-model}
\label{sec:bg:rust}

GPU programming splits host-side launch from device-side execution:
Rust's typed values collapse to raw pointers and scalars at the launch
boundary, and device kernels run many SIMT threads whose address
disjointness and visibility are programmer-managed invariants. \sysname
keeps CUDA's host launch and device execution structure but raises the
device program to the tile level: The programmer writes sequential code
over multi-dimensional tiles, and the compilation stack maps tile operations to
thread blocks, manages shared memory, and performs the parallel
decomposition. Element-wise kernels, matrix multiplication, reductions,
normalization, and attention are naturally tile-shaped. Schedules
outside the safe surface are built from the unsafe operations of
\S\ref{sec:design:unsafe}. The tile model's safety payoff is that the
operations a kernel performs are loads, stores, matrix multiplies, and
reductions over whole tiles, whose memory footprints the compiler can
see. There
is no per-thread shared-memory protocol for it to reason about.

A thread block is a unit of work scheduled onto a GPU streaming
multiprocessor. CUDA does not guarantee the order in which blocks run
or that they are resident at the same time. Its ordinary programming
model therefore requires blocks to execute independently, in any order,
in parallel or in series~\cite{cuda_programming_guide}. Threads within a
block can cooperate through shared memory and barriers; coordination
across blocks requires additional mechanisms or separate kernel launches.
To avoid that coordination, programmers structure many high-performance
tensor kernels so that each block has exclusive access to the regions
it mutates and shared read access to the remaining inputs.

The tile model exposes each block's computation as a single-threaded
program. At that level, the memory partitioning encouraged by CUDA's
execution model is exactly the aliasing-XOR-mutability discipline that
Rust's ownership rules enforce. The partitioned form that safe
\sysname requires is therefore the form these high-performance kernels
already take. With tensors partitioned on the host, \sysname can check
ownership of mutable output regions statically.

\label{sec:bg:tile-ir}
The order of a program's own loads and stores is the other half of the
safety argument, and Tile~IR~\cite{cuda_tile} leaves it to the program, exposing
memory ordering through \emph{tokens}: A token produced by one memory
operation and consumed by another orders the two, and operations without
token dependencies may be reordered.
\sysname threads tokens through mutable tensor operations, so exclusive
accesses carry ordering through the mutable view while shared reads
remain free for the compiler to reorder (\S\ref{sec:design:tokens}).

%% file: figures/src/compilation_pipeline.tex
%
%

\begin{tikzpicture}[
    >=Latex,
    thin,
    every node/.style={font=\footnotesize},
    box/.style={draw=black!70, rounded corners=1pt, fill=black!3,
                minimum height=0.40cm, inner sep=3pt, align=center},
    code/.style={box, font=\footnotesize\ttfamily},
    art/.style={box, minimum width=1.75cm, minimum height=0.36cm, inner sep=2pt},
    lab/.style={font=\footnotesize\itshape, text=black!55},
    flow/.style={->, black!75},
  ]

  \node[code] (source) {kernel.rs};
  \node[code, below=0.20cm of source] (pm) {\#[\cratename::module]};
  \draw[flow] (source) -- (pm);

  \node[art, below=0.32cm of pm]        (launcher) {kernel launcher};
  \node[art, left=0.15cm of launcher]   (desugar)  {desugared Rust};
  \node[art, right=0.15cm of launcher]  (ast)      {kernel AST};
  \draw[flow] (pm.south) -- (desugar.north);
  \draw[flow] (pm.south) -- (launcher.north);
  \draw[flow] (pm.south) -- (ast.north);

  \node[box, below=0.32cm of launcher] (hostbin) {host binary};
  \draw[flow] (desugar.south)  -- (hostbin.north west);
  \draw[flow] (launcher.south) -- (hostbin.north)
    node[midway, right=1.5pt, lab] {rustc};
  \draw[flow] (ast.south)      -- (hostbin.north east);

  \draw[decorate, decoration={brace, amplitude=2pt}]
    ([yshift=1.2mm]desugar.north west |- source.north) --
    ([yshift=1.2mm]ast.north east |- source.north)
    node[midway, above=0.5mm, lab] {Rust compile time};

  \node[box, below=0.42cm of hostbin] (ast_rt) {kernel AST};
  \node[box, right=0.45cm of ast_rt]  (tileir) {Tile IR};
  \node[box, right=0.45cm of tileir]  (cubin)  {cubin};
  \draw[->, dashed, black!55] (hostbin.south) -- (ast_rt.north);
  \draw[flow] (ast_rt) -- (tileir);
  \draw[flow] (tileir) -- (cubin);

  \draw[decorate, decoration={brace, mirror, amplitude=2pt}]
    ([yshift=-1.2mm]ast_rt.south west) --
    ([yshift=-1.2mm]cubin.south east)
    node[midway, below=0.5mm, lab] {runtime JIT loads embedded AST};

\end{tikzpicture}

%% file: sections/design.tex
\section{Safety Model}
\label{sec:design}
\label{sec:design:unsafe}

Our safety model extends Rust's across the launch boundary and into
the kernel. Host code is ordinary safe Rust and keeps Rust's own
guarantees. Inside the kernel, partitioning gives each tile program a
disjoint mutable view and token ordering gives it happens-before over
its own mutable operations. At the launch boundary, the generated
interfaces preserve ownership as host types become kernel parameters. \S\ref{sec:design:race} states the data-race-freedom theorem this
structure supports. Leaving the safe surface is an explicit choice at
one of two tiers: An \lstinline|unsafe fn| entry point may take raw
pointers or disable checks with \lstinline|unchecked_accesses|, making
the whole kernel unsafe code, or a safe kernel may perform individual
raw-pointer loads and stores in unsafe blocks, keeping every other
access checked. In both, \lstinline|unsafe| has the same semantics as in
Rust: The programmer upholds the invariants the compiler no longer
checks, and the low-level operations it exposes are the base for new
safe abstractions.

\subsection{Tile Programming Model}
\label{sec:design:overview}
\label{sec:design:h2d}

\begin{table}[t]
\centering
\footnotesize
\caption{Types at the kernel launch boundary: the host type, the kernel
  parameter it becomes, and the semantics the borrow checker enforces.
  The grammar below is the legend. \lstinline|P|, \lstinline|MP|, and
  \lstinline|MT| abbreviate \lstinline|Partition|,
  \lstinline|MappedLaunchPartition<Partition<..>>|, and
  \lstinline|MappedPartitionMut|; $M$ is the partition map of
  \S\ref{sec:design:escape}.}
\label{tab:type_mapping}
\setlength{\tabcolsep}{3.4pt}
\begin{tabular}{@{}lll@{}}
\textbf{Host type} & \textbf{Kernel parameter} & \textbf{Semantics} \\
\midrule
\lstinline|P<Tensor<T>>|           & \lstinline|&mut Tensor<T, S>| & move; exclusive \\
\lstinline|P<&mut Tensor>|         & \lstinline|&mut Tensor<T, S>| & borrow; exclusive \\
\lstinline|MP<Tensor<T>>|          & \lstinline|MT<T, S, M>|       & move; mapped excl. \\
\lstinline|MP<&mut Tensor>|        & \lstinline|MT<T, S, M>|       & borrow; mapped excl. \\
\lstinline|Tensor<T>|              & \lstinline|&Tensor<T, S>|     & move; shared ref. \\
\lstinline|Arc<Tensor<T>>|         & \lstinline|&Tensor<T, S>|     & shared; immutable \\
\lstinline|&Tensor<T>|             & \lstinline|&Tensor<T, S>|     & borrow; non-\lstinline|'static| \\
\lstinline|&TensorView<T>|         & \lstinline|&Tensor<T, S>|     & borrow; 0-copy view \\
\lstinline|T|                      & \lstinline|T|                 & by value; scalar \\
\lstinline|*const T|, \lstinline|*mut T| & unchanged               & \lstinline|unsafe| only; raw pointer \\
\midrule
\multicolumn{3}{@{}l@{}}{\textbf{Interface grammar}} \\[2pt]
\multicolumn{3}{@{}l@{}}{$T ::= $ \texttt{bool} $\mid$ \texttt{int} $\mid$ \texttt{float} \qquad $D ::= n \in \mathbb{Z}^{+} \mid$ \texttt{-1}} \\
\multicolumn{3}{@{}l@{}}{$S, M ::= [D_0, D_1, \ldots, D_{N-1}] \mid [D\texttt{;}\,N]$, $N \in \mathbb{Z}_{\ge 0}$} \\
\multicolumn{3}{@{}l@{}}{$\mathit{Param} ::= $ \texttt{\&mut Tensor<T, S>} $\mid$ \texttt{\&Tensor<T, S>} $\mid$ \texttt{MT<T, S, M>}} \\
\multicolumn{3}{@{}l@{}}{\phantom{$\mathit{Param} ::= $} $\mid$ \texttt{T} $\mid$ \texttt{*const T} $\mid$ \texttt{*mut T}} \\
\end{tabular}
\end{table}

A \sysname program is an ordinary Rust program whose GPU kernels live in
a \emph{tile module} (\lstinline|#[cutile::module]|). Functions marked
\lstinline|#[cutile::entry()]| are \emph{entry points}, the only
functions invocable from host code, and \sysname generates a typed
launch interface for each. A kernel body loads tiles from its tensor
parameters, computes on them (tiles are immutable values), and stores tiles through its
exclusive output parameter. Kernel parameters take one of the six forms of
Table~\ref{tab:type_mapping}. Tensor writes go through a partition because
partitioning is what gives each program a disjoint sub-tensor
(\S\ref{sec:design:escape}). Read-only accesses do not conflict with one
another, so tensors that the kernel does not mutate arrive as shared
references. Mutable tensors are partitioned on the host and arrive in
the kernel as exclusive \lstinline|&mut Tensor<T, S>| views over one
sub-tensor per program, or as the mapped partitions of
\S\ref{sec:design:escape}. Immutable tensors, whether owned, borrowed,
or behind an \lstinline|Arc|, arrive as shared
\lstinline|&Tensor<T, S>| views of the whole tensor. The proc macro
checks these mappings at compile time, and the launcher returns every
host value in the type it
was passed: An \lstinline|Arc<Tensor<T>>| comes back as an
\lstinline|Arc<Tensor<T>>|, a borrow as the borrow, so a launch reads
like an ordinary call that borrows or moves its arguments
(Figure~\ref{fig:async_lifecycle}\AppRuntime{; Appendix~\ref{app:runtime}}).

Host tensors are dynamically shaped, since their extents come from data
and configuration at run time. Kernel parameters carry a shape in their
type, and the launcher checks the one against the other. The shape
\lstinline|S| is a \emph{const generic array}: an array of dimensions
carried as a type parameter, each a static extent or a dynamic
\texttt{-1}, and any array expression over const generics (e.g.\
\lstinline|{[BM, BN]}| in Listing~\ref{lst:persistent_gemm_safe}) can be
written wherever a type parameter is expected. Stable Rust does not
support array-valued const parameters, so the module macro lowers each
shape parameter to a rank-specialized \emph{shadow type} with one scalar
const generic per dimension. Rank-polymorphic operations become traits with
one implementation per rank, which \texttt{rustc} resolves as usual.
Shapes are therefore checked with ordinary const generics:
\lstinline|mma| takes \lstinline|Tile<T, {[M, K]}>| and
\lstinline|Tile<T, {[K, N]}>| and returns \lstinline|Tile<T, {[M, N]}>|,
so a mismatched inner dimension is a type error. These shapes are the
concrete variadic types Tile~IR expects, and the JIT, which works from
the original source, constructs them directly.

\subsection{Partitioning}
\label{sec:design:partitioning}
\label{sec:design:escape}

A tile program's mutable tensors are partitioned before launch. In
Listing~\ref{lst:vec_add}, the host splits the 1024-element output into
128-element sub-tensors with \lstinline|.partition([128])|. On the
device, the generated entry function\AppRuntime{ (Appendix~\ref{app:runtime})} hands
each tile program an exclusive \lstinline|&mut Tensor| over one
sub-tensor, while the immutable inputs \lstinline|x| and \lstinline|y|
become shared \lstinline|&Tensor| inputs visible to all programs. Both
are Tile~IR views. A view is a structured window into GPU memory with a
shape and an element type. The exclusive one is a \emph{partition view}
and the shared one a \emph{tensor view}. Rust uses the same idiom for CPU parallelism: Rayon gives each closure
one disjoint chunk of a mutable slice, and \lstinline|partition| gives
each tile program one disjoint sub-tensor.

The launcher derives the launch grid from the partition shape.  CUDA grids are three-dimensional, so mutable tensors are capped at
rank 3, while immutable tensors are not. Safe code constructs
mutable partitions only through the partition APIs, from an owned
tensor (\lstinline|Partition<Tensor<T>>|, moved into the launch) or an
exclusive borrow (\lstinline|Partition<&mut Tensor<T>>|). Either way,
host access is unavailable until execution completes.

\begin{figure}[t]
  \centering
  \input{figures/src/partition_mapping}
  \caption{Host-to-device partitioning: a \lstinline|Partition<Tensor<T>>|
    gives each tile program its own exclusive sub-tensor (rose). A
    \lstinline|&Tensor<T>| is shared by every program as a read-only view (green).}
  \label{fig:partition_mapping}
\end{figure}

A \emph{mapped partition} specifies a \emph{partition map}, a function
chosen by the programmer from output sub-tensors to the tile programs
that own them. The plain \lstinline|&mut Tensor<T, S>| is the special case of one sub-tensor
per program. \lstinline|MappedPartitionMut<T, S, M>| lets the map
\lstinline|M| assign several sub-tensors to one program, which then owns
a bounded \emph{sequence} of them (Figure~\ref{fig:partition_map}).
The mapped form lowers to the same Tile~IR partition view and adds only
the obligation that a program write only sub-tensors assigned to it by
the map. The launcher checks the map's program count against
the number of sub-tensors. On the host, the mapped form is the plain partition plus
one call:

\noindent\begin{minipage}{\columnwidth}
\begin{lstlisting}[frame=none, numbers=none, xleftmargin=1em, aboveskip=3pt, belowskip=3pt]
// one sub-tensor per program (Listing 1)
let z = api::zeros::<f32>([M, N]).partition([BM, BN]);
// row i -> program i, four persistent programs
let z = api::zeros::<f32>([M, N]).partition([BM, BN])
  .map([4, 1], 4);
\end{lstlisting}
\end{minipage}

\noindent The second arrives in the kernel as
\lstinline|MappedPartitionMut<f32, {[BM, BN]}, {[4, 1]}>|.

A mapped partition lets GEMM express an efficient schedule. With one
program per output sub-tensor, each program reloads its share of the
$K$-dimension operands. Fewer, persistent programs, each walking a
sequence of sub-tensors, reuse operands across them. In
the map of Figure~\ref{fig:partition_map}, the four programs walk one
column together and share its \lstinline|y| tile, the schedule of the
GEMM of \S\ref{sec:placement} (Listing~\ref{lst:persistent_gemm_safe}). There,
\lstinline|MAP| instantiates the partition map and
\lstinline|z.iter_indices()| yields compiler-branded
\lstinline|PartitionIndex| values for the sub-tensors mapped to the
executing program~\cite{launchbury1994state,kiselyov2007capabilities,yanovski2021ghostcell}. A store checks that its index came from \lstinline|z|'s map before
lowering to Tile~IR, so the write stays safe although one program writes
many sub-tensors.

\begin{figure}[t]
  \centering
  \input{figures/src/partition_map_grid}
  \caption{The partition map on a $4{\times}4$ grid of output
    sub-tensors. Left: the plain partition, one program per sub-tensor.
    Right: the map \texttt{[4, 1]} with four persistent programs, row $i$
    to program $i$, visited left to right (step $k$). Each sub-tensor
    still has exactly one owner.}
  \label{fig:partition_map}
\end{figure}

The general object behind this API is the \emph{partition mapping
function}: Any function from output sub-tensors to tile programs yields
disjoint ownership, because each sub-tensor has exactly one image. An implementation must supply each program's preimage as the sequence
of sub-tensors it iterates. \lstinline|MappedPartitionMut| restricts
these functions to the family described by \lstinline|M|, the grouped
schedules that GEMM-like kernels use for operand reuse, and its iterator
enumerates each program's preimage directly, so the bounds of every
store are proven at compile time. That family covers 22 of the 25 kernels of the Qwen3
inference engine of \S\ref{sec:eval:grout}. The remaining three express
schedules outside it with six one-line unsafe blocks inside otherwise
safe kernels. A user-defined mapping function from a tensor partition to tile programs
could be evaluated on the host before launch to produce the same
preimages. Generalizing the iterator machinery to it is future work
(\S\ref{sec:conclusion}).

\subsection{Data-Race Freedom}
\label{sec:design:race}

The store below is from a kernel that permutes an attention tensor from
shape $(B, H, M, D)$ to $(B, M, H D)$, merging the head and feature axes
into the destination's third dimension. The store is shown at both tiers:

\noindent\begin{minipage}{\columnwidth}
\begin{lstlisting}[frame=none, numbers=none, xleftmargin=1em, aboveskip=3pt, belowskip=3pt]
// indexed store: the programmer carries the obligation
unsafe { dst.store_unchecked(tile, [b, m, 0]) }; // BUG
// safe store: the target is the partition view itself
dst.store(tile);
\end{lstlisting}
\end{minipage}

\noindent The first form addresses the destination by index. It is an
unsafe store, and the programmer carries the obligation that no two tile
programs write to the same location. The correct tile index is
\lstinline|[b, m, h]|. This instance omits the head
coordinate, so the $H$ programs that share $(b, m)$ each load a
different source tile and write it to the same destination. Nothing
orders these writes (\S\ref{sec:bg:tile-ir}), so under the Tile~IR
memory model, the program has a data race and undefined behavior. Over ten runs, 7.7--8.7\% of output elements differ between
consecutive runs (61--70\% in the contended tiles) and 98--99\% differ
from the intended result. The indexed store is the store
interface of mainstream tile DSLs~\cite{tillet2019triton, cutile_python},
and nothing in those DSLs rules this bug class out, because the program chooses the destination.

\sysname makes this store impossible to construct: \lstinline|dst| is a
partition view (\lstinline|&mut Tensor<f32, {[1, BM, BD]}>|) over the
destination, each tile program owns one sub-tensor, and it stores to the
view itself, with no index to omit or mix up.

\label{sec:design:tokens}
Disjoint partitions rule out conflicts between programs. Within a
program, mutable operations must also keep their order, and Tile~IR's
token-ordered operations have no default program order
(\S\ref{sec:bg:tile-ir}), so the compiler threads a token chain through
the mutable tensor operations it emits. For a kernel that loads from and
then stores to a mutable tensor~\lstinline|c|:

\begin{center}
\small
$t_0 \xrightarrow{\texttt{c.load()}} t_1 \xrightarrow{\texttt{c.store()}} t_2$
\end{center}

\noindent The store consumes the token the load produced, so it observes
the load. Loads through immutable references carry no token, and the Tile
compiler reorders them freely. Token threading plays the role that
sequenced-before plays for ordinary Rust.

Every execution of a kernel written in safe \sysname is data-race-free
under Tile~IR's memory model, and the property extends across launches
to safe host programs. Two tensor accesses conflict when they access the
same memory location and at least one is a write. In safe code, that write
goes through the issuing program's
partition view. An access from another program cannot touch that view,
since partitions are disjoint. An access from the same program through an
immutable view cannot touch the partition either, since immutable views
cannot alias partitions. The remaining case, two accesses through the
program's own \lstinline|&mut| view, is ordered by the token chain, so
the pair is morally strong under Tile~IR's definition. Across launches,
the host execution model preserves operand lifetimes and access rights
until completion. Conflicting accesses across launches follow stream order,
and dependencies are recorded during graph capture to preserve that order
on replay.
This establishes happens-before where the API permits reuse
(Theorem~\ref{thm:drf}, Appendix~\ref{app:drf}).

%% file: figures/src/partition_mapping.tex
%

\begin{tikzpicture}[
    >=Latex,
    thin,
    every node/.style={font=\footnotesize},
    hostbox/.style={draw=black!75, rounded corners=1.5pt,
                    minimum width=1.9cm, minimum height=0.38cm,
                    inner sep=1.5pt, font=\footnotesize\ttfamily},
    thead/.style={font=\footnotesize\bfseries},
    cell/.style={draw=black!70, rounded corners=1pt, fill=ownMut,
                 minimum width=1.05cm, minimum height=0.38cm,
                 inner sep=1.5pt, font=\footnotesize\ttfamily, align=center},
  ]

  \node[cell] (c0) {\&mut T$_0$};
  \node[cell, right=0.16cm of c0] (c1) {\&mut T$_1$};
  \node[right=0.16cm of c1] (cd) {$\cdots$};
  \node[cell, right=0.16cm of cd] (cn) {\&mut T$_{N{-}1}$};

  \node[thead, above=0.10cm of c0] (h0) {$P_0$};
  \node[thead, above=0.10cm of c1] (h1) {$P_1$};
  \node[thead, above=0.10cm of cn] (hn) {$P_{N-1}$};
  \draw[decorate, decoration={brace, amplitude=2pt}]
    ([yshift=1.0mm]h0.north west) -- ([yshift=1.0mm]hn.north east)
    node[midway, above=0.5mm, font=\footnotesize\itshape, text=black!55, align=center]
    {tile programs $P_i$, one tile block each\\sub-tensor $i \mapsto P_i$};

  \node[hostbox, anchor=east] (part) at ($(c0.west) + (-0.32,0)$)
    {Partition<Tensor<T>{>}};
  \draw[->] (part.east) -- (c0.west);

  \path
    let \p1=(c0.south west), \p2=(cn.south east) in
    coordinate (band_nw) at (\x1, \y1 - 0.14cm)
    coordinate (band_se) at (\x2, \y1 - 0.50cm);
  \path[draw=black!70, rounded corners=1pt, fill=ownShared]
    (band_nw) rectangle (band_se);
  \node[font=\footnotesize\ttfamily] at ($(band_nw)!0.5!(band_se)$)
    {\&Tensor<T, S>};
  \coordinate (band_west_mid) at ($(band_nw)!0.5!(band_nw |- band_se)$);
  \node[hostbox, anchor=east] (ref) at ($(band_west_mid) + (-0.32,0)$)
    {\&Tensor<T>};
  \draw[->] (ref.east) -- (band_west_mid);

\end{tikzpicture}

%% file: figures/src/partition_map_grid.tex
\begin{tikzpicture}[
    every node/.style={font=\footnotesize},
    cell/.style={draw=black!70, rounded corners=1pt, minimum width=0.62cm,
                 minimum height=0.38cm, inner sep=1pt, font=\footnotesize\ttfamily,
                 align=center},
    lab/.style={font=\footnotesize\itshape, text=black!55, align=center},
  ]
  \colorlet{p0}{ownMut}
  \colorlet{p1}{ownP1}
  \colorlet{p2}{ownP2}
  \colorlet{p3}{ownP3}
  \foreach \i in {0,...,3} {
    \foreach \j in {0,...,3} {
      \pgfmathtruncatemacro{\pid}{4*\i+\j}
      \node[cell, fill=ownMut] (L\i\j) at (0.66*\j, -0.42*\i) {p\pid};
    }
  }
  \node[lab, anchor=south] at ($(L00.north west)!0.5!(L03.north east) + (0,0.08)$)
    {\texttt{\&mut Tensor}\\sub-tensor $i \mapsto$ program $i$};
  \foreach \i in {0,...,3} {
    \foreach \j in {0,...,3} {
      \node[cell, fill=p\i] (R\i\j) at (4.2 + 0.66*\j, -0.42*\i)
        {p\i\,\textsuperscript{\j}};
    }
  }
  \node[lab, anchor=south] at ($(R00.north west)!0.5!(R03.north east) + (0,0.08)$)
    {\texttt{MappedPartitionMut}, \texttt{MAP = [4, 1]}\\row $i \mapsto$ program $i$};
  \draw[-{Latex[length=1.4mm]}, black!55] ([yshift=-0.10cm]R30.south west) -- ([yshift=-0.10cm]R33.south east)
    node[midway, below=0.2mm, lab] {step $k$: the four programs\\walk column $k$ together};
\end{tikzpicture}

%% file: sections/placement.tex
\section{Bounds-Check Placement}
\label{sec:placement}


Every load from and store to a partition in a kernel carries a proof
obligation per axis: The coordinate must be non-negative and must select a tile
inside the tensor, $0 \le i$ and $i \cdot t < e$ for tile size $t$ and
extent $e$. These are tile-index bounds: an edge tile may extend past the
tensor. Tile~IR masks stores for those elements; loads use the view's
padding value, or unspecified values if none is set~\cite{cuda_tile}.
Read-only \lstinline|Tensor::partition| uses zero padding.
A check that discharges this obligation at the access site may sit
in the hot loop of the kernel, on the path between the loads and the
tensor-core operation, where it consumes registers and instruction
slots. On the two kernels of \S\ref{sec:eval:checks}, checks left at
their access sites slow the kernel by 17--22\%. The compiler lowers that cost by \emph{bounds-check placement}, code
motion for in-kernel checks whose targets include the host. Check placement
resolves each obligation once during \emph{specialization}, the step of
JIT compilation that fixes the kernel to the generic parameters of its
first launch (e.g. tile shapes) while dynamic extents remain launch-time
values (Figure~\ref{fig:compilation_pipeline}\AppRuntime{; Appendix~\ref{app:runtime}}).
Those generic parameters are the compiler's \emph{assumptions}. When the
assumptions cannot discharge an obligation, placement emits a check at the
least costly point where all required values are available.

Resolution has four outcomes (Figure~\ref{fig:check_ladder}). An
obligation is \emph{discharged at specialization} when it follows from
the assumptions, and no runtime check is emitted anywhere. Checks are discharged in
three cases: the index is the iterand of the axis it indexes, as the $K$ loop of
Listing~\ref{lst:persistent_gemm_safe} iterates
\lstinline|0..num_tiles(&part_x, 1)|; a literal coordinate is within the bounds of a static extent; or a mapped-partition store uses a branded index, which
is in bounds by construction (\S\ref{sec:design:escape}).

An obligation becomes a \emph{launch precondition} when every operand is
known at launch (e.g.\ tensor extents and tile counts). The check moves
into the generated launcher, runs once per launch before any work is
queued, and returns an error on failure. The cross-tensor tie of
Listing~\ref{lst:persistent_gemm_safe},
$\lceil \mathit{dim}(x,1)/BK \rceil \le \lceil \mathit{dim}(y,0)/BK
\rceil$, is one such check.

An obligation is hoisted to a \emph{loop preheader} when it cannot become
a launch precondition, because one of its operands exists only inside
the kernel, but its check can still leave the loop. Such an operand is a
value the kernel computes (e.g.\ a coordinate decoded, inside a persistent loop, from the partition
index of the current sub-tensor). An index that is affine in the loop's
iterand turns the per-iteration check into one inequality over the
loop's range. A loop-invariant index is
tested directly. The check moves outward to the preheader of the
outermost loop at which all of its operands exist, guarded so that a loop
with no iterations cannot trap.

An obligation stays \emph{in place}, checked at the access site against
the actual index, when the access may not execute on every iteration, under a
runtime conditional or in a loop body with an early exit, or when our
analysis cannot bound a non-affine index. The entry attribute
\lstinline|deny_in_kernel_checks| turns the last two outcomes into errors
during specialization, so a kernel
that compiles with it contains no bounds-check assert.

\begin{figure}[t]
  \centering
  \input{figures/src/check_ladder}
  \caption{Bounds-check placement: resolution sites, the kernel census
    (Appendix~\ref{app:census}), and a worked GEMM example. $K_x$ and
    $K_y$ are the operands' reduction-axis tile counts. The loop bound
    proves $0\le k<K_x$, eliminating the X check; the Y check becomes
    one host launch precondition.
    Neither check remains in the GPU loop. The pseudocode omits other
    axes and computation.}
  \label{fig:check_ladder}
\end{figure}

\begin{figure}[tb]
\begin{lstlisting}[caption={Persistent GEMM over a mapped output partition.
  \texttt{z.iter\_indices()}
  yields the branded \texttt{PartitionIndex} values assigned to this
  program. The kernel carries \texttt{deny\_in\_kernel\_checks}.},
  label={lst:persistent_gemm_safe}]
#[cutile::entry(deny_in_kernel_checks = true)]
fn gemm_persistent<T: ElementType, const BM: i32,
    const BN: i32, const BK: i32, const MAP: [i32; 2]>(
  mut z: MappedPartitionMut<T, {[BM, BN]}, MAP>,
  x: &Tensor<T, {[-1, -1]}>,
  y: &Tensor<T, {[-1, -1]}>,
) {
  let part_x = x.partition(shape![BM, BK]);
  let part_y = y.partition(shape![BK, BN]);

  for out_idx in z.iter_indices() {
    let (bid_m, bid_n) = out_idx.components();
    let mut tz = constant(T::ZERO, shape![BM, BN]);
    for k_tile in 0i32..num_tiles(&part_x, 1) {
      let tx = part_x.load([bid_m, k_tile]);
      let ty = part_y.load([k_tile, bid_n]);
      tz = mma(tx, ty, tz);
    }
    z.store(tz, out_idx);
  }
}
\end{lstlisting}
\end{figure}

Listing~\ref{lst:persistent_gemm_safe} carries four obligations, two
loads over two axes each, and none reach the kernel\AppCensusGemm{ (Table~\ref{tab:census_gemm})}. The $K$-loop index into \lstinline|x| is
discharged at specialization, since the loop's range is that axis's
tile count. The other three become launch preconditions, each
rejected on the host with its own message when violated.
The kernel compiles under \lstinline|deny_in_kernel_checks|, and its Tile~IR contains no asserts.

The causal prefill attention kernel of \S\ref{sec:eval:checks} exercises
the other outcomes (Figure~\ref{fig:check_ladder}). Of its nine
obligations, three are discharged at specialization, two become launch
preconditions, and four are hoisted to the preheader of its inner loop,
where they run once per work step of the persistent loop instead of on
each of the loop's 256 iterations. Our current analysis leaves those four
at the preheader, because their operands are decoded inside the
persistent loop from the
partition index that \lstinline|iter_indices| yields for the current
sub-tensor. Appendix~\ref{app:census} lists every obligation against
the kernel's source.

Obligations whose operands are decoded from that partition index, such
as attention's four, stop at the preheader today. Because
\lstinline|iter_indices| covers a map known at launch and the operand
chains are monotone in its coordinates, a universally quantified launch check over the map space could decide
them before launch. That would be a launch precondition of a new form,
and we leave it to future work. We return to the other limits of placement in \S\ref{sec:limitations}.

%% file: figures/src/check_ladder.tex
\definecolor{plJit}{HTML}{76B900}
\definecolor{plLaunch}{HTML}{1DBBA4}
\definecolor{plPre}{HTML}{E69F00}
\definecolor{plSite}{HTML}{D55E00}
\begin{tikzpicture}[font=\footnotesize,
  station/.style={draw, rounded corners=1pt, minimum height=6.5mm,
                  text width=18mm, align=center, inner sep=1.5pt},
  mv/.style={-{Latex[length=1.6mm]}, line width=0.5pt}]
\node[station, fill=plJit!22]    (jit)    {discharged at\\specialization};
\node[station, fill=plLaunch!22, right=2.5mm of jit]    (launch) {launch\\precondition};
\node[station, fill=plPre!25, right=2.5mm of launch] (pre)    {loop\\preheader};
\node[station, fill=plSite!20, right=2.5mm of pre]    (site)   {in\\place};
\draw[mv] (site) -- (pre);
\draw[mv] (pre) -- (launch);
\draw[mv] (launch) -- (jit);
\node[above=1mm of site.north east, anchor=south east, font=\footnotesize\itshape]
  {a check may move left as far as its operands exist};
\draw[decorate, decoration={brace, mirror, amplitude=2pt}]
  ([yshift=-1.2mm]jit.south west) -- ([yshift=-1.2mm]launch.south east)
  node[midway, below=0.8mm] {host, before any GPU work};
\draw[decorate, decoration={brace, mirror, amplitude=2pt}]
  ([yshift=-1.2mm]pre.south west) -- ([yshift=-1.2mm]site.south east)
  node[midway, below=0.8mm] {device};
\end{tikzpicture}
\par\vspace{2mm}
\setlength{\fboxsep}{0pt}\setlength{\tabcolsep}{3.5pt}
\begin{tabular}{@{}lccccc@{}}
\footnotesize kernel & \footnotesize obligations & \footnotesize discharged & \footnotesize launch & \footnotesize preheader & \footnotesize in place \\
\hline
\footnotesize GEMM & \footnotesize 4 & \colorbox{plJit!25}{\footnotesize\strut\,1\,} & \colorbox{plLaunch!25}{\footnotesize\strut\,3\,} & \colorbox{plPre!28}{\footnotesize\strut\,0\,} & \colorbox{plSite!22}{\footnotesize\strut\,0\,} \\
\footnotesize attention & \footnotesize 9 & \colorbox{plJit!25}{\footnotesize\strut\,3\,} & \colorbox{plLaunch!25}{\footnotesize\strut\,2\,} & \colorbox{plPre!28}{\footnotesize\strut\,4\,} & \colorbox{plSite!22}{\footnotesize\strut\,0\,} \\
\end{tabular}
\par\bigskip
\input{figures/src/check_placement_example}

%% file: figures/src/check_placement_example.tex
\begin{tikzpicture}[
  font=\footnotesize,
  code/.style={font=\footnotesize\ttfamily, anchor=west, inner sep=1pt},
  heading/.style={font=\footnotesize\bfseries, anchor=west},
  region/.style={font=\footnotesize\itshape, text=black!65, anchor=west},
  move/.style={-{Latex[length=1.5mm]}, line width=0.6pt, draw=plLaunch!65!black},
]
  \draw[black!40, line width=0.4pt] (0,0.65) -- (8.15,0.65);
  \node[fill=white, text=black!65, font=\footnotesize\itshape,
    inner xsep=4pt, inner ysep=0pt] at (4.075,0.65) {GEMM example};
  \node[heading] at (0,0) {Before placement};
  \node[heading] at (4.35,0) {After placement};

  \node[region] at (0,-0.38) {Host};
  \node[region] at (4.35,-0.38) {Host};
  \node[code, fill=plLaunch!22] (launch) at (4.35,-0.75)
    {check(Kx <= Ky);};
  \draw[black!35, dashed] (0,-1.02) -- (3.80,-1.02);
  \draw[black!35, dashed] (4.35,-1.02) -- (8.15,-1.02);

  \node[region] at (0,-1.20) {Device};
  \node[region] at (4.35,-1.20) {Device};
  \node[code] at (0,-1.58) {for k in 0..Kx \{};
  \node[code, fill=plJit!22] at (0.18,-1.94) {assert(k < Kx);};
  \node[code, fill=plLaunch!22] (check) at (0.18,-2.30)
    {assert(k < Ky);};
  \node[code] at (0.18,-2.66) {load\_x(k);};
  \node[code] at (0.18,-3.02) {load\_y(k);};
  \node[code] at (0,-3.38) {\}};

  \node[code] at (4.35,-1.58) {for k in 0..Kx \{};
  \node[code] at (4.53,-1.94) {load\_x(k);};
  \node[code] at (4.53,-2.30) {load\_y(k);};
  \node[code] at (4.35,-2.66) {\}};

  \draw[move] (check.east) -- (3.65,-2.30) -- (3.65,-0.75)
    -- (launch.west);
\end{tikzpicture}

%% file: sections/execution.tex
\section{Execution Model}
\label{sec:execution}
\label{sec:design:execution}

On the host side, efficient GPU work scheduling is asynchronous. The
type system therefore has to keep host access and in-flight device
writes apart without forcing the host to block.

\begin{figure}[t]
  \centering
  \input{figures/src/async_lifecycle}
  \caption{An awaited device operation returns the operands it was given.
    A poll spins on the stream within the budget and otherwise parks the
    future on the reactor.}
  \label{fig:async_lifecycle}
\end{figure}

\input{figures/src/safety_overhead}

Host-side GPU work is built around \lstinline|DeviceOp|, a trait for
lazy, typed GPU work that composes with other operations before anything
is submitted to CUDA\@. Like Rust's iterators, it builds work
statically. The generated launcher accepts either a host value or a
device operation for each parameter and returns a device operation, so a
kernel call whose operands are device operations is itself a device
operation. Combinators such as \lstinline|then|, which chains a dependent
operation, compose in the same way. The result is a nested operation
type that forms a directed acyclic graph of operations. Nothing runs until a method such as
\lstinline|.sync()| executes it. A generated launcher's \lstinline|execute| impl runs the operations that
produce the kernel's arguments, launches the kernel on the stream of its
execution context, and returns its arguments\AppRuntime{ (Appendix~\ref{app:runtime})}. Composition is efficient because each operation in the graph is
submitted to the stream while its predecessors are still in flight. The returned
tensors may therefore name memory whose writes have not completed, so
\lstinline|execute| is unsafe. The \lstinline|sync| method executes the
graph of operations safely by synchronizing the stream before it returns
the arguments. In pseudocode:\AppRuntime{ Appendix~\ref{app:runtime} gives the trait, a
generated implementation, and the combinators.}

\noindent\begin{minipage}{\columnwidth}
\begin{lstlisting}[frame=none, numbers=none, xleftmargin=1em, aboveskip=3pt, belowskip=3pt]
trait DeviceOp {
  type Out: Send;
  unsafe fn execute(self, cx: &Ctx) -> Result<Self::Out>;
  fn sync(self) -> Result<Self::Out> {  // safe
    let cx = Ctx::new(stream()?);
    let out = unsafe { self.execute(&cx) };
    cx.stream().synchronize()?; out
  }
}
\end{lstlisting}
\end{minipage}

A \lstinline|DeviceOp| executes in one of three modes:
\lstinline|.sync()| blocks until the work completes; \lstinline|.await|
(via \lstinline|IntoFuture|) suspends the current task until it completes;
and \lstinline|.graph()| captures the composition as a CUDA graph for
low-latency replay of work launched repeatedly. A replay of the captured
graph is itself a device operation, so it executes synchronously or
asynchronously like any other.

\subsection{Asynchronous Execution}
\label{sec:execution:async}
Every device operation implements \lstinline|IntoFuture|. A future holds an operation's operands until an \lstinline|.await|
returns them (Figure~\ref{fig:async_lifecycle}). A kernel that
outlives the block that launched it, as a spawned task does, must own its
arguments or hold \lstinline|'static| references, as any async runtime's
spawn requires. The asynchronous form costs the programmer one method on the same
device operation:

\noindent\begin{minipage}{\columnwidth}
\begin{lstlisting}[frame=none, numbers=none, xleftmargin=1em, aboveskip=3pt, belowskip=3pt]
let (z, x, y) = add(z, x, y).sync()?;   // blocks
let (z, x, y) = add(z, x, y).await?;    // suspends
let t = spawn(add(z, x, y).into_future()); // awaitable
\end{lstlisting}
\end{minipage}

Short GPU work leaves asynchronous execution little room for host
overhead.
Decode-step kernels run for one to a few tens of microseconds
(Appendix~\ref{app:exp3_extras}) while CUDA's host callback for completion
notification costs tens of microseconds per operation. To reduce completion
overhead, we combine two techniques used in microsecond-scale systems:
polling and event-driven
completion~\cite{barroso2017killer, kalia2019erpc}. The future first polls
the stream for a bounded time, then falls back to a completion reactor.
The default budget of
$20~\mu$s is set an order of magnitude above the reactor's wake cost and
above the duration of most decode-step kernels, so most waits resolve at
the poll site. The reactor is one thread that wakes a registered future once its stream writes a
completion flag into host memory. Short waits therefore avoid a reactor
wakeup, while longer waits release the host thread after the spin budget.
The roughly $2~\mu$s wake cost is then amortized over the longer GPU
execution (Appendix~\ref{app:exp3_extras}).
Async matches sync across the evaluated pipeline lengths
(\S\ref{sec:eval:execution}, Figure~\ref{fig:exec_mode_scaling}). \AppRuntime{ Appendix~\ref{app:runtime} gives the completion path in detail.}

\subsection{CUDA Graphs}
\label{sec:execution:graphs}
\lstinline|CudaGraph::scope| provides a closure-based API over a scope
handle \lstinline|s|. Each \lstinline|s.record(op)| adds the operation \lstinline|op| to the
graph as a node. Recording a node does not execute its GPU work.
Releasing a borrow after \lstinline|record()| therefore cannot race with
the recorded operation. The graph preserves capture order on replay,
so buffers may be reused between nodes.
A \lstinline|GraphNode| marker trait
restricts capture to non-allocating operations, so replayed nodes do not
depend on addresses that change across replays.

%% file: figures/src/async_lifecycle.tex
\begin{tikzpicture}[
    >=Latex, thin,
    every node/.style={font=\footnotesize},
    box/.style={draw=black!70, rounded corners=1pt, fill=black!3,
                minimum height=0.44cm, inner sep=3pt, align=center},
    test/.style={draw=black!70, fill=white, diamond, aspect=2.4,
                 inner sep=1pt, align=center},
    code/.style={box, fill=white, font=\footnotesize\ttfamily},
    lab/.style={font=\footnotesize\itshape, text=black!55},
    stp/.style={font=\footnotesize\ttfamily, text=black!60},
    flow/.style={->, black!75},
  ]
  \node[box]                         (poll)   {poll};
  \node[box,  right=0.55cm of poll]  (launch) {launch};
  \node[test, right=0.6cm of launch] (chk)    {done?};
  \node[box,  right=0.6cm of chk]    (ret)    {return operands};
  \node[box,  below=0.4cm of chk]    (park)   {park on reactor};
  \node[box,  below=0.4cm of ret]    (wake)   {wake};
  \draw[flow] (poll)   -- (launch);
  \draw[flow] (launch) -- (chk)  node[midway, above=1pt, lab] {spin};
  \draw[flow] (chk)    -- (ret)  node[pos=0.3, above=1pt, lab] {yes};
  \draw[flow] (chk)    -- (park) node[midway, right=2pt, lab] {no};
  \draw[flow] (park)   -- (wake) node[midway, above=1pt, lab] {flag};
  \draw[flow] (wake)   -- (ret);
  \node[draw=black!45, rounded corners=2pt, inner sep=0.2cm,
        fit=(poll)(launch)(chk)(ret)(park)(wake)] (rt) {};
  \node[lab, anchor=south west] at ([xshift=2pt, yshift=1pt]rt.south west)
    {async runtime};
  \node[code, anchor=south west] (op)
    at ([yshift=0.5cm]rt.north west) {add(z, x, y).await?};
  \node[code, anchor=south] (res)
    at ([yshift=0.5cm]rt.north -| ret) {(z, x, y)};
  \draw[flow] (poll |- op.south) -- (poll.north)
    node[pos=0.3, right=2pt, stp] {IntoFuture};
  \draw[flow] (ret.north) -- (res.south)
    node[pos=0.7, right=2pt, lab] {same types};
  \draw[->, densely dotted, black!55] (op) -- (res)
    node[midway, above=1pt, lab] {returns};
\end{tikzpicture}

%% file: figures/src/safety_overhead.tex
\begin{figure*}[t]
  \centering
  \begin{subfigure}[t]{0.29\textwidth}
    \centering
    \ifexpRatioPlots\input{figures/tikz/exp1_elemwise_ratio}\else\input{figures/tikz/exp1_elemwise}\fi
    \subcaption{Element-wise add (memory-bound).}
    \label{fig:elemwise_safety}
  \end{subfigure}
  \hfill
  \begin{subfigure}[t]{0.39\textwidth}
    \centering
    \ifexpRatioPlots\input{figures/tikz/exp1_gemm_ratio}\else\input{figures/tikz/exp1_safety_overhead}\fi
    \subcaption{GEMM (compute-bound).}
    \label{fig:gemm_safety}
  \end{subfigure}
  \hfill
  \begin{subfigure}[t]{0.27\textwidth}
    \centering
    \input{figures/tikz/exp1_jit_breakdown}
    \subcaption{JIT overhead.}
    \label{fig:jit_breakdown}
  \end{subfigure}
  \caption{Checked versus unchecked kernels, and JIT overhead.
    \ifexpRatioPlots
    Upper strips: safe Rust as a fraction of peak, cuBLAS against the same
    peak. Lower strips: safe Rust and cuTile Python relative to unsafe
    Rust. Whiskers are the numerator's spread ($p_{25}$--$p_{75}$
    element-wise, ${\pm}1$~sd GEMM), the dotted line parity.
    \else
    ``Python'' and ``Rust'' are cuTile Python and \sysname. Dotted
    lines show the peak. Error bars show $p_{25}$--$p_{75}$ for
    element-wise and time standard deviation as TFlop/s for GEMM.
    \fi}
  \label{fig:safety_overhead}
\end{figure*}

%% file: figures/tikz/exp1_elemwise_ratio.tex
\definecolor{exPython}{HTML}{E69F00}
\definecolor{exUnsafe}{HTML}{009E73}
\definecolor{exSafe}{HTML}{0072B2}
\definecolor{exPeak}{HTML}{888888}
\begin{tikzpicture}[baseline=(group c1r1.north)]
\begin{groupplot}[
  group style={group size=1 by 2, vertical sep=9pt, xticklabels at=edge bottom},
  scale only axis,
  width=4.0cm,
  xmode=log,
  log basis x=2,
  xtick={1048576,2097152,4194304,8388608,16777216,33554432,67108864,134217728,268435456},
  xticklabels={$2^{20}$,$2^{21}$,$2^{22}$,$2^{23}$,$2^{24}$,$2^{25}$,$2^{26}$,$2^{27}$,$2^{28}$},
  axis lines*=left,
  axis line style={line width=0.5pt},
  ymajorgrids=true,
  grid style={black!25, line width=0.3pt},
  every axis plot/.append style={line width=0.9pt},
  legend cell align=left,
  legend style={draw=none, fill=none},
  legend image post style={mark size=1.6pt},
  error bars/y dir=both,
  error bars/y explicit,
  error bars/error bar style={line width=0.4pt},
  error bars/error mark options={rotate=90, mark size=0.9pt, line width=0.4pt},
  ylabel style={font=\footnotesize, at={(axis description cs:-0.15,0.5)}, anchor=south},
]
\nextgroupplot[height=1.45cm, ymin=0, ymax=100, ytick={0,50,100},
  ylabel={\% of peak}, xticklabels={}, xlabel={}]
\addplot [exSafe, mark=*, mark size=1.4pt] table [x=n, y=safe_pct_peak] {figures/tikz_data/elemwise_ratio.dat};
\nextgroupplot[height=1.25cm, ymin=0.8, ymax=1.1, ytick={0.8,0.9,1,1.1},
  ylabel={vs.\ unsafe}, xlabel={$N$ (elements)}]
\addplot [black!45, densely dotted, line width=0.6pt, mark=none, forget plot]
  coordinates {(1048576,1) (268435456,1)};
\addplot [exSafe, mark=*, mark size=1.4pt]
  table [x=n, y=safe, y error plus expr=\thisrow{safe_hi}-\thisrow{safe}, y error minus expr=\thisrow{safe}-\thisrow{safe_lo}] {figures/tikz_data/elemwise_ratio.dat};
\addplot [exPython, mark=square*, mark size=1.5pt]
  table [x=n, y=py, y error plus expr=\thisrow{py_hi}-\thisrow{py}, y error minus expr=\thisrow{py}-\thisrow{py_lo}] {figures/tikz_data/elemwise_ratio.dat};
\end{groupplot}
\end{tikzpicture}

%% file: figures/tikz/exp1_elemwise.tex
\input{figures/tikz_data/exp1_params}
\begin{tikzpicture}
\definecolor{exPython}{HTML}{E69F00}
\definecolor{exUnsafe}{HTML}{009E73}
\definecolor{exSafe}{HTML}{0072B2}
\definecolor{exPeak}{HTML}{888888}
\begin{axis}[
  width=\linewidth,
  height=0.56\linewidth,
  axis lines*=left,
  xmode=log,
  log basis x=2,
  xmin=964689.92,  
  xmax=289910292.5, 
  ymin=\expElemYmin,
  ymax=\expElemYmax,
  xtick={1048576,2097152,4194304,8388608,16777216,33554432,67108864,134217728,268435456},
  xticklabels={$2^{20}$,$2^{21}$,$2^{22}$,$2^{23}$,$2^{24}$,$2^{25}$,$2^{26}$,$2^{27}$,$2^{28}$},
  ytick={2,4,6,8},
  xlabel={$N$ (elements)},
  ylabel={TB/s (f16)},
  ymajorgrids=true,
  grid style={black!25, line width=0.3pt},
  axis line style={line width=0.5pt},
  line width=0.9pt,
  legend cell align=left,
  legend style={
    at={(0.97,0.05)}, anchor=south east, row sep=-3pt,
    draw=none, fill=none,
    inner sep=1.5pt,
    /tikz/every even column/.append style={column sep=2pt},
  },
  legend image post style={mark size=1.6pt},
  every axis plot/.append style={
    error bars/y dir=both,
    error bars/y explicit,
    error bars/error bar style={line width=0.4pt},
    error bars/error mark options={rotate=90, mark size=0.9pt, line width=0.4pt},
  },
]
\addplot [exPeak, densely dotted, line width=0.6pt, mark=none]
  coordinates {(964689.92,\expElemPeak) (289910292.5,\expElemPeak)};
\addlegendentry{Peak}
\addplot [exPython, mark=square*, mark size=1.5pt]
  table [x=n, y=tbs, y error plus=eplus, y error minus=eminus]
  {figures/tikz_data/elemwise_python.dat};
\addlegendentry{Python}
\addplot [exUnsafe, mark=triangle*, mark size=1.8pt]
  table [x=n, y=tbs, y error plus=eplus, y error minus=eminus]
  {figures/tikz_data/elemwise_optimized.dat};
\addlegendentry{unsafe Rust}
\addplot [exSafe, mark=*, mark size=1.4pt]
  table [x=n, y=tbs, y error plus=eplus, y error minus=eminus]
  {figures/tikz_data/elemwise_safe.dat};
\addlegendentry{Rust}
\end{axis}
\end{tikzpicture}

%% file: figures/tikz/exp1_gemm_ratio.tex
\definecolor{exPython}{HTML}{E69F00}
\definecolor{exUnsafe}{HTML}{009E73}
\definecolor{exSafe}{HTML}{0072B2}
\definecolor{exPeak}{HTML}{888888}
\begin{tikzpicture}[baseline=(group c1r1.north)]
\begin{groupplot}[
  group style={group size=1 by 2, vertical sep=9pt, xticklabels at=edge bottom},
  scale only axis,
  width=4.0cm,
  xmode=log,
  log basis x=2,
  xtick={1024,2048,4096,8192,16384,32768},
  xticklabels={$2^{10}$,$2^{11}$,$2^{12}$,$2^{13}$,$2^{14}$,$2^{15}$},
  axis lines*=left,
  axis line style={line width=0.5pt},
  ymajorgrids=true,
  grid style={black!25, line width=0.3pt},
  every axis plot/.append style={line width=0.9pt},
  legend cell align=left,
  legend style={draw=none, fill=none},
  legend image post style={mark size=1.6pt},
  error bars/y dir=both,
  error bars/y explicit,
  error bars/error bar style={line width=0.4pt},
  error bars/error mark options={rotate=90, mark size=0.9pt, line width=0.4pt},
  ylabel style={font=\footnotesize, at={(axis description cs:-0.15,0.5)}, anchor=south},
]
\nextgroupplot[height=1.45cm, ymin=0, ymax=100, ytick={0,50,100},
  ylabel={\% of peak}, xticklabels={}, xlabel={},
  legend columns=1,
  legend style={font=\footnotesize, draw=none, fill=none, inner sep=1pt, row sep=-2pt,
                at={(1.05,-0.04)}, anchor=west},
  legend image post style={mark size=1.2pt}]
\addplot [exSafe, mark=*, mark size=1.4pt] table [x=m, y=safe_pct_peak] {figures/tikz_data/gemm_ratio.dat};
\addlegendentry{rust}
\addlegendimage{exPython, mark=square*, mark size=1.5pt}
\addlegendentry{python}
\addplot [black, dashed, mark=x, mark size=1.8pt, mark options={solid}] table [x=m, y=cublas_pct_peak] {figures/tikz_data/gemm_ratio.dat};
\addlegendentry{cublas}
\nextgroupplot[height=1.25cm, ymin=0.8, ymax=1.1, ytick={0.8,0.9,1,1.1},
  ylabel={vs.\ unsafe}, xlabel={$M = N = K$}]
\addplot [black!45, densely dotted, line width=0.6pt, mark=none, forget plot]
  coordinates {(1024,1) (32768,1)};
\addplot [exSafe, mark=*, mark size=1.4pt]
  table [x=m, y=safe, y error plus expr=\thisrow{safe_hi}-\thisrow{safe}, y error minus expr=\thisrow{safe}-\thisrow{safe_lo}] {figures/tikz_data/gemm_ratio.dat};
\addplot [exPython, mark=square*, mark size=1.5pt]
  table [x=m, y=py, y error plus expr=\thisrow{py_hi}-\thisrow{py}, y error minus expr=\thisrow{py}-\thisrow{py_lo}] {figures/tikz_data/gemm_ratio.dat};
\end{groupplot}
\end{tikzpicture}

%% file: figures/tikz/exp1_safety_overhead.tex
\input{figures/tikz_data/exp1_params}
\begin{tikzpicture}
\definecolor{exPython}{HTML}{E69F00}
\definecolor{exUnsafe}{HTML}{009E73}
\definecolor{exSafe}{HTML}{0072B2}
\definecolor{exPeak}{HTML}{888888}
\begin{axis}[
  width=\linewidth,
  height=0.56\linewidth,
  axis lines*=left,
  xmode=log,
  log basis x=2,
  xmin=942.08,   
  xmax=35389.44, 
  ymin=0,
  ymax=\expGemmYmax,
  xtick={1024,2048,4096,8192,16384,32768},
  xticklabels={$2^{10}$,$2^{11}$,$2^{12}$,$2^{13}$,$2^{14}$,$2^{15}$},
  ytick={0,500,1000,1500,2000},
  xlabel={$M = N = K$},
  ylabel={TFlops (f16)},
  yticklabel style={/pgf/number format/1000 sep={}},
  ymajorgrids=true,
  grid style={black!25, line width=0.3pt},
  axis line style={line width=0.5pt},
  line width=0.9pt,
  legend cell align=left,
  legend style={
    at={(0.97,0.05)}, anchor=south east,
    row sep=-3pt,
    draw=none, fill=none,
    inner sep=1.5pt,
  },
  legend image post style={mark size=1.6pt},
  every axis plot/.append style={
    error bars/y dir=both,
    error bars/y explicit,
    error bars/error bar style={line width=0.4pt},
    error bars/error mark options={rotate=90, mark size=0.9pt, line width=0.4pt},
  },
]
\addplot [exPeak, densely dotted, line width=0.6pt, mark=none]
  coordinates {(942.08,\expGemmPeak) (35389.44,\expGemmPeak)};
\addlegendentry{Peak}
\addplot [exPython, mark=square*, mark size=1.5pt]
  table [x=m, y=tflops, y error plus=eplus, y error minus=eminus]
  {figures/tikz_data/gemm_python.dat};
\addlegendentry{Python}
\addplot [exSafe, mark=*, mark size=1.4pt]
  table [x=m, y=tflops, y error plus=eplus, y error minus=eminus]
  {figures/tikz_data/gemm_rust_safe.dat};
\addlegendentry{Rust}
\addplot [exUnsafe, mark=triangle*, mark size=1.8pt]
  table [x=m, y=tflops, y error plus=eplus, y error minus=eminus]
  {figures/tikz_data/gemm_rust_unsafe.dat};
\addlegendentry{unsafe Rust}
\addplot [black, dashed, mark=x, mark size=1.8pt, mark options={solid}]
  table [x=m, y=tflops, y error plus=eplus, y error minus=eminus]
  {figures/tikz_data/gemm_cublas.dat};
\addlegendentry{cuBLAS}
\end{axis}
\end{tikzpicture}

%% file: figures/tikz/exp1_jit_breakdown.tex
\input{figures/tikz_data/exp1_params}
\begin{tikzpicture}[baseline=(current axis.north)]
\definecolor{jitBackend}{HTML}{4C72B0}
\definecolor{jitFrontend}{HTML}{DD8452}
\definecolor{jitModload}{HTML}{55A868}
\definecolor{jitEdge}{HTML}{222222}
\begin{axis}[
  scale only axis,
  width=0.30\linewidth,
  height=3.02cm,  
  axis lines*=left,
  ybar stacked,
  bar width=0.55,
  xmin=-0.7,
  xmax=1.7,
  ymin=0,
  ymax=900,
  xtick={0,1},
  xticklabels={128,256},
  ytick={0,200,400,600,800},
  xlabel={BM},
  ylabel={JIT overhead (ms)},
  ymajorgrids=true,
  grid style={black!25, line width=0.25pt},
  axis on top=false,
  axis line style={line width=0.5pt},
  legend cell align=left,
  legend style={
    at={(1.06,0.5)}, anchor=west,
    draw=none, fill=none,
    inner sep=1pt,
    /tikz/every even column/.append style={column sep=1pt},
  },
  legend image code/.code={
    \draw [#1, draw=jitEdge, line width=0.3pt]
      (0cm,-0.06cm) rectangle (0.18cm,0.12cm);
  },
]
\addplot [fill=jitBackend, draw=jitEdge, line width=0.3pt]
  table [x expr=\coordindex, y=backend] {figures/tikz_data/jit_breakdown.dat};
\addlegendentry{back-end}
\addplot [fill=jitFrontend, draw=jitEdge, line width=0.3pt]
  table [x expr=\coordindex, y=frontend] {figures/tikz_data/jit_breakdown.dat};
\addlegendentry{front-end}
\addplot [fill=jitModload, draw=jitEdge, line width=0.3pt]
  table [x expr=\coordindex, y=moduleload] {figures/tikz_data/jit_breakdown.dat};
\addlegendentry{module-load}
\end{axis}
\end{tikzpicture}

%% file: sections/evaluation.tex
\section{Evaluation}
\label{sec:evaluation}

We measure the runtime cost of safety (\S\ref{sec:eval:safety}), the
cost that bounds-check placement recovers (\S\ref{sec:eval:checks}), the
host overhead of each execution mode (\S\ref{sec:eval:execution}), and
the end-to-end performance of an inference engine built on \sysname
(\S\ref{sec:eval:grout}). Experiments run on one GPU of an NVIDIA
DGX~B200 or on a workstation with an NVIDIA GeForce RTX~5090. Each
experiment states which. All measurements use default GPU clocks and report medians over
repeated samples. These timings predate the host-runtime safety fixes.

\subsection{Checked vs.\ Unchecked Kernels}
\label{sec:eval:safety}

To measure what safety costs on the GPU, we use a compute-bound and a
memory-bound workload on the B200: GEMM and element-wise add. These runs
use CUDA~13.3 and driver~595.58.03, with the benchmark process pinned to
logical CPU~0.

For GEMM (Figure~\ref{fig:gemm_safety}), we compare Rust, unsafe Rust,
cuTile Python, and cuBLAS over $M{=}N{=}K$ in powers
of two from 1024 to 32768 at \texttt{f16} precision, with tile shapes
tuned once per size and shared across Rust and Python (as an autotuner
would). Rust is the mapped-partition kernel of
Listing~\ref{lst:persistent_gemm_safe}. Unsafe Rust runs the same
schedule with unchecked raw-pointer output access and a hand-written
persistent loop. cuTile Python runs the same schedule on the same Tile~IR backend. cuBLAS is the fastest of the algorithms cuBLASLt's heuristics propose
at each size. Unsafe Rust on the same backend isolates the cost of the safety
machinery (a comparison with Triton~\cite{tillet2019triton} would
measure backend code quality instead). cuBLAS and the hardware peak
anchor the results in absolute terms. At $M{=}N{=}K{=}8192$, Rust reaches 2.11~PFlop/s, 94\% of the
B200's 2.25~PFlop/s dense \texttt{f16} peak and 98.0\% of cuBLAS. Unsafe Rust matches it within 0.1\%. The safe partition-index
formulation therefore imposes no measurable cost over manual pointer
access. cuTile Python reaches 2.06~PFlop/s (96.0\% of cuBLAS), which
places both frontends on the same backend performance envelope. Safe and
unsafe Rust agree within 0.2\% from $4096^3$ up. At $1024^3$ and
$2048^3$, where the kernels run for 9 and $17~\mu$s, the safe kernel
trails by 3\% and 1\%. cuTile Python, whose launch path costs more per
call, trails Rust by 9\% at $1024^3$ (2--3\% at element-wise sizes up to
$2^{22}$) and by 1--4\% from $2048^3$ up.

Element-wise add (Figure~\ref{fig:elemwise_safety}) isolates memory-bound tile
loads and stores. At
$N=2^{28}$, unsafe Rust and Rust both reach 6.97~TB/s and cuTile Python reaches
6.96~TB/s, near the 7.68~TB/s peak.

Figure~\ref{fig:jit_breakdown} reports the first-use compilation cost
separate from these steady-state results and shows it growing with tile
size because larger GEMM tiles produce larger Tile~IR kernels. An
application that installs the on-disk kernel cache pays this cost once
per machine.

\subsection{Bounds-Check Placement}
\label{sec:eval:checks}

\begin{figure}[t]
  \centering
  \input{figures/tikz/exp4_check_ablation}
  \caption{Bounds-check placement ablation, normalized to the unchecked
    arm. Under each kernel, its census from \S\ref{sec:placement}:
    discharged / launch / preheader / in place. Whiskers are $p_{25}$--$p_{75}$
    over rounds.}
  \label{fig:check_ablation}
\end{figure}

\S\ref{sec:placement} claims that the cost of bounds safety depends on
where each check runs. To measure that dependence, we compare three
builds of the same kernel: \emph{unchecked} compiles the accesses
without checks; \emph{unoptimized checks} evaluates every obligation at
its access site, with discharge, launch preconditions, and hoisting
disabled; and \emph{optimized checks} is the shipped placement. The kernels are the
persistent GEMM of Listing~\ref{lst:persistent_gemm_safe} (\texttt{f16},
$M{=}N{=}K{=}8192$, $128{\times}128{\times}64$ tiles) and the causal
prefill attention kernel of \S\ref{sec:placement} (\texttt{f16}, $q_{\mathrm{len}}{=}8192$,
$32{\times}32$ tiles, Listing~\ref{lst:attention_excerpt}). The arms are interleaved in paired rounds on the RTX~5090 at default
clocks, 11 rounds for GEMM and 9 for attention. Each round is normalized to its unchecked arm. All three arms produce
bitwise-identical outputs.

Unoptimized checks cost $1.17\times$ on GEMM and $1.22\times$ on
attention (Figure~\ref{fig:check_ablation}). Optimized placement runs at $1.00\times$ and $1.02\times$. After
placement, the GEMM kernel contains no check, and attention's four
remaining checks run in the preheader of its inner loop, once per loop
entry instead of once per iteration.

\subsection{Sync, Async, and Graph Execution}
\label{sec:eval:execution}

We next measure how host-side launch overhead scales with the length
$N$ of a \emph{pipeline} under the sync, async, and graph modes. A
pipeline is a chain of $N$ operations. $N$ runs from 1 to 1000 on the RTX~5090 over a small element-wise
$y = x \cdot g$ kernel on a length-2048 f16 tile, so launch overhead
dominates. The pipeline-length sweep pins the benchmark process to
logical CPU~8.

Sync runs a
\lstinline|.then()| chain of $N$ launches with one \lstinline|.sync()|,
which isolates launch cost from synchronization cost. Async runs the same
chain through \lstinline|.await|, yielding the host task while the GPU
runs, with the staged completion path of \S\ref{sec:execution:async}.
Graph captures the pipeline and replays all $N$ kernels with one driver
call.

Figure~\ref{fig:exec_mode_scaling} separates two host-cost regimes.
Synchronizing after every kernel, with individual
\lstinline|.sync_on(stream)| calls, costs 6 to $7~\mu$s per operation at
every length and is left off the plot. Sync and async pay one launch per
kernel and one wait per pipeline. With spin-polled completion they track
each other at every pipeline length and converge to about $1.7~\mu$s per
operation (Appendix~\ref{app:exp3_extras}). Graph replay removes the
host-side per-kernel launch overhead and approaches the rate at which the GPU itself dispatches kernels, about
$0.8~\mu$s per operation on this hardware. The capture line is a graph built and replayed once. Graphs are meant
for repeated replay of one sequence, so a single replay never repays the
construction cost of about $1.4~\mu$s per node.

\begin{figure}[t]
  \centering
  {\footnotesize\pgfplotslegendfromname{execmodelegend}\par}
  \begin{subfigure}[t]{0.555\columnwidth}
    \centering
    \input{figures/tikz/exp2_execmode_latency}
    \subcaption{Per-op latency: lines are medians, bands $p_{25}$--$p_{75}$. Capture is construction plus one replay.}
    \label{fig:exec_mode_scaling}
  \end{subfigure}\hfill
  \begin{subfigure}[t]{0.415\columnwidth}
    \centering
    \input{figures/tikz/exp2_execmode_grout}
    \subcaption{Decode-step replay: medians with
      $p_{25}$--$p_{75}$ whiskers.}
    \label{fig:exec_mode_grout}
  \end{subfigure}
  \caption{Sync, async, and graph execution (f16, RTX~5090). (b) replays
    synthetic decode steps built from the Qwen3-4B (437 kernels) and
    Qwen3-32B (838 kernels) profiles.}
  \label{fig:exec_mode}
\end{figure}

Figure~\ref{fig:exec_mode_grout} measures the same execution modes on
synthetic decode steps calibrated from the inference engine's profiles. From
kernel-level profiles of the Qwen3-4B and Qwen3-32B decode steps
(\S\ref{sec:eval:grout}), each kernel class's per-step launch count and
median duration drive a chain of calibrated element-wise kernels, which run with cold caches as their real counterparts do. Sync, async, and graph
replay agree within 2\% on both steps. Each step waits once at its end
like an autoregressive loop. Within a step, the host submits launches
faster than the GPU executes them and all three modes are bound by GPU
time. Synchronizing after
every kernel (not shown) is 33--38\% slower. Graph capture is a one-time cost of 1.1--2.1~ms. At this scale the choice among sync,
async, and graph is one of ergonomics and amortization.

Async yields the host thread once the spin budget expires. This helps when the host has concurrent work to overlap, as in tool calling,
where the tool runs on a core that would otherwise be occupied by a
synchronous wait, or interruptible voice-to-text, where the host prepares
the next audio window while the device processes the current one. \AppOverlap{Appendix~\ref{app:exp3_extras} quantifies the overlap and the
host-thread footprint.}

\subsection{Case Study: Qwen3 Inference Engine}
\label{sec:eval:grout}

\enginenamens\enginecite{} is an open-source, lean batch-1 Qwen3 engine
built on \sysname. Every kernel is written in \sysname except the dense projections,
which dispatch to cuBLAS, as in the engines we compare against. Of its 25 kernels, 22 are safe end to end, including all attention and fused-norm kernels. Migrating the engine's
twelve legacy unsafe kernels to their safe forms left prefill within
1\% and decode within 0.2\% of the unsafe versions on the RTX~5090, each
form at its own tuned tile shapes, with identical outputs. The remaining three build schedules outside the
mapped-partition family from six one-line unsafe blocks, each inside an
otherwise safe kernel (\S\ref{sec:design:unsafe}). Prefill runs through
a cached sequence of typed operations over a reusable tensor pool. The
decode pass consists of a fused normalization and cache-update kernel, decode
attention, and split-K merges. It is recorded once as
\lstinline|DeviceOp| graph nodes inside \lstinline|CudaGraph::scope| and
replayed for subsequent tokens.

\begin{table}[t]
  \caption{Qwen3 single-request performance (\texttt{f16}, batch~1,
    medians; prefix caching disabled): decode throughput at its peak over the generation sweep and at
    generation length $t_g{=}8192$ (prompt length $p_p{=}18$), and prefill
    latency at $p_p{=}8192$ ($t_g{=}36$). The full sweeps with spreads
    are in Appendix~\ref{app:grout_sweeps}.}
  \label{tab:grout}
  \centering
  \footnotesize
  \setlength{\tabcolsep}{4pt}
  \begin{tabular}{llrrr}
    \toprule
    & & \multicolumn{2}{c}{decode (tok/s)} & prefill (ms) \\
    & engine & peak & $t_g{=}8192$ & $p_p{=}8192$ \\
    \midrule
    RTX~5090 & \enginename & \textbf{170.9} & \textbf{154.7} & \textbf{318.5} \\
    Qwen3-4B & vLLM   & 164.3 & 151.3 & 404.6 \\
             & SGLang & 166.2 & 154.3 & 432.8 \\
    \midrule
    B200      & \enginename & \textbf{81.6} & \textbf{80.1} & 513.8 \\
    Qwen3-32B & vLLM   & 78.8 & 77.5 & \textbf{491.6} \\
              & SGLang & 77.7 & 76.5 & 507.1 \\
    \bottomrule
  \end{tabular}
\end{table}

\enginename performs on par with the two most widely used open-source
inference engines vLLM~\cite{kwon2023vllm} and
SGLang~\cite{zheng2024sglang} (Table~\ref{tab:grout}). We evaluate
their batch-1 performance. Its decode
throughput is at or above theirs at every point of the generation sweep
on both platforms, by up to 7\%. Its prefill latency is 11--28\% below
vLLM's on the RTX~5090. On the B200 it is 10--15\% below vLLM's at short
prompts, matches it at $p_p{=}512$ and 2048, and trails it by 4.5\% at
8192. The remaining gap at long prompts on the B200 comes mostly from the
prefill attention kernel. vLLM and SGLang run attention kernels
hand-tuned and precompiled for the B200; \enginename's is one kernel
whose source is the same on both GPUs. The difference is in the code the
backend generates on the B200's sm\_100 architecture.

The memory bandwidth roofline $R = \beta / (W + \bar{K})$ is the generated-token rate
when decode is bound only by device-memory bandwidth, for peak
device-memory bandwidth $\beta$, model weights $W$ in bytes, and average
KV-cache bytes $\bar{K}$ read per generated token. At $t_g{=}8192$,
decode runs at 75\% (RTX~5090) and 67\% (B200) of it. Every measured
engine sits in that regime (Appendix~\ref{app:grout_sweeps}).

%% file: figures/tikz/exp4_check_ablation.tex
\begin{tikzpicture}
\definecolor{exUnchecked}{HTML}{999999}
\definecolor{exWritten}{HTML}{D55E00}
\definecolor{exOptim}{HTML}{009E73}
\begin{axis}[
  width=0.64\linewidth,
  height=0.30\linewidth,
  enlarge x limits=0.30,
  axis lines*=left,
  ybar,
  bar width=9pt,
  ymin=0.9, ymax=1.3,
  symbolic x coords={gemm,attn},
  xtick=data,
  xticklabels={{GEMM\\1 / 3 / 0 / 0}, {attention\\3 / 2 / 4 / 0}},
  xticklabel style={align=center},
  ylabel={Normalized time},
  ymajorgrids=true,
  grid style={black!25, line width=0.3pt},
  axis line style={line width=0.5pt},
  legend style={draw=none, at={(1.03,0.5)}, anchor=west, legend columns=1},
]
\addplot[fill=exUnchecked] coordinates {(gemm,1.0) (attn,1.0)};
\addplot[fill=exWritten, error bars/.cd, y dir=both, y explicit,
  error bar style={black, line width=0.4pt}, error mark options={black, mark size=1.2pt}]
  coordinates {(gemm,1.1650) -= (0,0.0138) += (0,0.0060)
               (attn,1.2213) -= (0,0.0025) += (0,0.0065)};
\addplot[fill=exOptim, error bars/.cd, y dir=both, y explicit,
  error bar style={black, line width=0.4pt}, error mark options={black, mark size=1.2pt}]
  coordinates {(gemm,0.9980) -= (0,0.0102) += (0,0.0027)
               (attn,1.0193) -= (0,0.0013) += (0,0.0011)};
\legend{unchecked, unoptimized checks, optimized checks}
\end{axis}
\end{tikzpicture}

%% file: figures/tikz/exp2_execmode_latency.tex
\begin{tikzpicture}[baseline=(current axis.north)]
\begin{axis}[
  scale only axis,
  width=1.36in,
  height=0.75in,
  xmode=log, ymode=log,
  log ticks with fixed point,
  /pgf/number format/1000 sep={},
  xlabel={Pipeline length $N$},
  ylabel={Time per op ($\mu$s)},
  axis lines*=left,
  axis line style={black},
  tick style={black, thin},
  grid=both,
  major grid style={black!12, line width=0.3pt},
  minor grid style={black!8, line width=0.3pt},
  legend to name=execmodelegend,
  legend style={
    draw=none, fill=none,
    legend columns=4, column sep=0.6em,
  },
  legend cell align=left,
  legend image post style={mark size=1.6pt},
  legend image code/.code={
    \draw[mark repeat=2, mark phase=2, #1]
      plot coordinates {(0cm,0cm) (0.14cm,0cm) (0.28cm,0cm)};
  },
  ymin=0.62, ymax=16,
  ytick={1,2,5,10},
  minor ytick={3,4,6,7,8,9},
  xmin=0.85, xmax=1200,
  clip=false,
]
\addplot[name path=sc lo, draw=none, forget plot] table[x=n, y=p25] {figures/tikz_data/exp2a_sync-chained.dat};
\addplot[name path=sc hi, draw=none, forget plot] table[x=n, y=p75] {figures/tikz_data/exp2a_sync-chained.dat};
\addplot[expSyncCh, fill opacity=0.15, draw=none, forget plot] fill between[of=sc lo and sc hi];

\addplot[name path=as lo, draw=none, forget plot] table[x=n, y=p25] {figures/tikz_data/exp2a_async.dat};
\addplot[name path=as hi, draw=none, forget plot] table[x=n, y=p75] {figures/tikz_data/exp2a_async.dat};
\addplot[expAsync, fill opacity=0.15, draw=none, forget plot] fill between[of=as lo and as hi];

\addplot[name path=gr lo, draw=none, forget plot] table[x=n, y=p25] {figures/tikz_data/exp2a_graph.dat};
\addplot[name path=gr hi, draw=none, forget plot] table[x=n, y=p75] {figures/tikz_data/exp2a_graph.dat};
\addplot[expGraph, fill opacity=0.15, draw=none, forget plot] fill between[of=gr lo and gr hi];

\addplot[expSyncCh, line width=0.9pt, mark=square*, mark size=1.4pt]
  table[x=n, y=med] {figures/tikz_data/exp2a_sync-chained.dat};
\addlegendentry{sync}
\addplot[expAsync, line width=0.9pt, mark=triangle*, mark size=1.7pt]
  table[x=n, y=med] {figures/tikz_data/exp2a_async.dat};
\addlegendentry{async}
\addplot[expGraphCap, line width=0.9pt, dash pattern=on 2.4pt off 1.4pt,
  mark=diamond, mark size=1.7pt, mark options={solid}]
  table[x=n, y=med] {figures/tikz_data/exp2a_graph-first-use.dat};
\addlegendentry{capture}
\addplot[expGraph, line width=0.9pt, mark=diamond*, mark size=1.7pt]
  table[x=n, y=med] {figures/tikz_data/exp2a_graph.dat};
\addlegendentry{replay}
\end{axis}
\end{tikzpicture}

%% file: figures/tikz/exp2_execmode_grout.tex
\begin{tikzpicture}[baseline=(current axis.north)]
\begin{axis}[
  scale only axis,
  width=0.88in,
  height=0.75in,
  ybar,
  bar width=4.5pt,
  axis lines*=left,
  axis line style={black},
  tick style={black, thin},
  ymajorgrids,
  major grid style={black!12, line width=0.3pt},
  ylabel={Time per step (ms)},
  ymin=0, ymax=16,
  xmin=-0.45, xmax=1.45,
  xtick={0,1},
  xticklabels={4B, 32B},
  xlabel={Qwen3 decode step},
  xlabel style={yshift=4.5pt},  
  x tick label as interval=false,
  enlarge x limits=false,
  every error bar/.append style={black!60, line width=0.5pt},
  error bars/y explicit,
  error bars/error bar style={line width=0.5pt},
]
\addplot[ybar, fill=expSyncCh, draw=none,
  error bars/.cd, y dir=both, y explicit]
  table[x=x, y=med, y error plus=hi, y error minus=lo] {figures/tikz_data/exp2c_sync.dat};
\addplot[ybar, fill=expAsync, draw=none,
  error bars/.cd, y dir=both, y explicit]
  table[x=x, y=med, y error plus=hi, y error minus=lo] {figures/tikz_data/exp2c_async.dat};
\addplot[ybar, fill=expGraphCap, draw=none,
  error bars/.cd, y dir=both, y explicit]
  table[x=x, y=med, y error plus=hi, y error minus=lo] {figures/tikz_data/exp2c_graph-capture.dat};
\addplot[ybar, fill=expGraph, draw=none,
  error bars/.cd, y dir=both, y explicit]
  table[x=x, y=med, y error plus=hi, y error minus=lo] {figures/tikz_data/exp2c_graph-replay.dat};
\end{axis}
\end{tikzpicture}

%% file: sections/limitations.tex
\section{Limitations}
\label{sec:limitations}

The tile model supports sparse matrix multiplication, stencils, and
convolutions as well as the dense tensor kernels evaluated here. Tile~IR
controls the mapping to SIMT threads, shared-memory layout, and
thread-level synchronization. Code that requires explicit control of
these mechanisms is written in CUDA~C++ and composed through host-side FFI.

The current safe tensor API requires mutable outputs to be partitioned
among tile programs before launch. It does not provide shared atomic
access to tensor outputs, so shared-output reductions currently require
\lstinline|unsafe| atomic operations. Data-dependent scatters and
histograms also require \lstinline|unsafe| operations when their writes
cannot be confined to the assigned partitions. Additional safe tensor
types could support these patterns; for example, a type permitting shared
atomic updates could provide a safe interface for cross-program reductions.

Bounds-check placement moves a check only as far as its operands and the
control flow around it allow. A check under a runtime conditional, or in
a loop body with an early exit, stays at the access site, because the access
may not execute on every iteration. A check also stays there when our
analysis cannot bound a non-affine index. Checks whose operands are decoded from the partition index of the
current sub-tensor stop at the enclosing loop's preheader, and checks whose operands are loaded from device memory never
leave the kernel. A launch precondition is unconditional: it tests only
whether the check's operands are known at launch, not whether the access
would execute, so a launch can be rejected for a violation the kernel
would not have committed.

Three implementation limitations motivate the future work in
\S\ref{sec:conclusion}. The partition-map family covers
the grouped schedules of \S\ref{sec:design:escape}. Three of \enginename's
25 kernels use six one-line view constructions to express schedules
outside this family. A general mapping function would let these kernels
express their schedules without unsafe view construction.
A kernel body is the subset
of Rust that the JIT lowers from the source AST, type-checked by
\texttt{rustc} but not compiled by it. The shape types rely on the shadow
encoding of \S\ref{sec:design:h2d} until Rust supports array-valued const
parameters directly.

The evaluation has a fixed scope. Every inference measurement is a
single request at batch 1 in f16 on one model family, Qwen3, on two GPUs
at default clocks. The safe-versus-unsafe comparisons of
\S\ref{sec:eval:safety} and \S\ref{sec:eval:checks} cover GEMM,
element-wise add, and attention, and the engine-wide comparison of
\S\ref{sec:eval:grout} covers its remaining kernels once. Microbenchmarks
report medians over repeated samples, and the placement ablation reports
medians and interquartile ranges over paired rounds.

%% file: sections/related_work.tex
\section{Related Work}
\label{sec:related}

\paragraph{GPU programming in Rust.}
Rust-CUDA, rust-gpu, and cuda-oxide~\cite{rust_cuda, rust_gpu,
cuda_oxide} compile Rust to GPU device code and established Rust as a viable kernel
language, and GPU Offload in Rust~\cite{drehwald2026offload} integrates
GPU code generation into \texttt{rustc} through LLVM Offload and carries
ownership across host-device transfers. CubeCL~\cite{cubecl} adds a SIMT-style kernel
model across backends, and Candle and Burn~\cite{candle, burn} provide
safe host-side tensor APIs for ML applications. \sysname is
complementary to all of these: It preserves Rust's
\lstinline|&mut|/\lstinline|&| aliasing contract across the launch
boundary and into a tile-based kernel body (\S\ref{sec:design}), so the
safety argument covers the kernel as well as the host.
RustBelt~\cite{jung2018rustbelt} gives Rust's guarantees formal
foundations, Rudra~\cite{bae2021rudra} shows that \lstinline|unsafe|
Rust is an ecosystem-scale bug source, and Abdi et
al.~\cite{abdi2024fearless} ask when Rust parallelism is fearless and
zero-cost. \S\ref{sec:evaluation} answers that question for tile
kernels.

\paragraph{Tile-based GPU programming.}
Triton~\cite{tillet2019triton} popularized tile-level GPU programming. Pallas, ThunderKittens, and TensorIR \cite{pallas,
spector2024thunderkittens, feng2023tensorir} bring tile-level models to
JAX, CUDA~C++, and tensor compilers. CUDA Tile~C++
and cuTile Python~\cite{cutile_python, cuda_tile} share \sysname's
Tile~IR backend. In \sysname, tensor and partition types let Rust check
ownership across host code and GPU kernels.

\paragraph{Safe heterogeneous programming.}
Descend and Mojo \cite{descend2023, mojo} bring ownership to GPU
programming through new languages. Descend uses type-level tracking of
the thread hierarchy and memory regions for SIMT code. \sysname works
inside Rust's existing type system and targets the tile model, where
the ownership invariant is tractable to check.
GPUVerify~\cite{betts2012gpuverify} verifies kernels after the fact,
Faial~\cite{cogumbreiro2021faial} checks data-race freedom of SIMT
kernels compositionally, and
Regent, Halide, Futhark, and Prism~\cite{slaughter2015regent,
ragan2013halide, henriksen2017futhark, bansal2025prism} structure
parallelism through regions, schedules, purity, or thread granularity. \sysname constrains which memory a program may access and in what
order, by construction in the launch and tile APIs.

\paragraph{Deferred execution.}
Task-based runtimes such as StarPU, PaRSEC, and
Legion~\cite{augonnet2011starpu, bosilca2013parsec, bauer2012legion}
defer GPU work behind runtime dependency tracking over registered data.
\lstinline|DeviceOp| expresses dependencies through typed operation
composition, with Rust checking ownership of the operands.

\paragraph{Bounds-check placement.}
Bounds-check elimination in optimizing and JIT compilers proves checks
away or hoists them within the compiled program by range analysis and
loop-invariant code motion~\cite{gupta1990bounds, bodik2000abcd,
wurthinger2007abce}. The placement of \S\ref{sec:placement} applies the
same reasoning at kernel specialization time, with the launch
parameters known, and adds a tier outside the kernel: a launch
precondition evaluated once per launch.

%% file: sections/conclusion.tex
\section{Conclusion}
\label{sec:conclusion}

\sysname extends Rust's ownership discipline across the launch boundary
and into GPU kernels: Mutable parameters become disjoint partitioned
accesses, immutable references become shared tensor views, and the
generated launch interface carries that contract between host and
device. This contract supports the data-race-freedom theorem.
Bounds-check placement brings the evaluated checked kernels to parity
with their unchecked versions. Asynchronous execution matches
synchronous execution in the pipeline benchmarks, and graph replay
reduces per-launch cost. On GEMM, the safe kernel
matches a raw-pointer variant and reaches 98\% of cuBLAS
on the B200. \enginename, a Qwen3 inference engine \ifanon{}{developed
in collaboration with Hugging Face and }built on \sysname, reaches
single-request throughput on par with vLLM and SGLang on the RTX~5090
and the B200.

The future direction of this work is a safe SIMT model that composes
with tile programming, reusing the same tensor types, partitioning, and
launch contract across arbitrary levels of the thread hierarchy. SIMT code is
hard to make safe statically because hundreds of threads in a thread
block share memory and synchronize explicitly; safe code must know which threads execute
an instruction together and which values are uniform across them. A
second direction is the general partition mapping function of
\S\ref{sec:design:escape}, which would cover the three \enginename
kernels that still carry unsafe blocks. Two
enabling directions lie in Rust itself: the
stabilization of array-valued const parameters for shape types; and lowering kernels from rustc's mid-level IR (MIR) once it is stable, to
enable the complete Rust language for \sysname's tile programming model.

%% file: sections/acknowledgements.tex
\begin{acks}
We thank Gonzalo Brito for design discussions and support for tile programming
in Rust, Simon Cooksey (Tile IR memory model author) for confirming that our
data-race-freedom argument is correct,
Jason Knight for support and direction, Nihal Pasham for
early contributions, and the PSA and ARG research groups and the Tile~IR and
CUDA teams for support. We also thank Jed Brown and Sarah El Kazdadi for discussions on safe GPU
programming in Rust.
\end{acks}

%% file: sections/appendix.tex
\appendix

\section{Data-Race Freedom}
\label{app:drf}

\sysname's safe tensor API wraps Tile~IR's load and store operations. Beyond
data-race freedom, those operations carry only Tile~IR's \emph{validity}
requirement, in-bounds access to live, initialized memory, which is the ordinary
memory-safety obligation of the tensor API. This section discharges the one
requirement specific to Tile~IR's weakly ordered memory model, data-race
freedom, for any kernel and launch written in safe Rust. The unsafe tiers of
\S\ref{sec:design:unsafe} (\lstinline|unchecked_accesses| and raw pointers)
carry the requirement manually and are excluded. The argument's trusted base is
the proc-macro-generated launcher and kernel entry, which realize the disjoint
partitions and token threading of \S\ref{sec:design:partitioning} and
\S\ref{sec:design:tokens}; the runtime
JIT, which correctly lowers the kernel to Tile~IR; the Tile~IR backend,
which honors the token ordering it is given; the tensor and execution
runtimes, which preserve storage validity and access rights; and CUDA's stream and graph
execution ordering. We assume these components and do not prove them.

\paragraph{Tile~IR memory model.}
We use Tile~IR's definitions~\cite{cuda_tile}, whose \emph{morally strong}
relation follows the formalized PTX memory model~\cite{lustig2019ptx}. A kernel runs a grid of tile
programs (\S\ref{sec:bg:gpu-programming-model}), each a single logical thread
that executes as one \emph{tile block}, the granularity at which the model
orders memory. Two accesses \emph{conflict} when they touch the same location
and at least one is a write. Ordering is carried by tokens: An operation
\emph{waits-for} another when it depends on a token the other produced, and
\emph{restricted program order} is the intersection of program order and this
token order. Two same-location accesses are \emph{morally strong} when they are
related in restricted program order, or each names a scope that contains the
other's tile block. A \emph{data race} is a conflicting pair that is not related by
happens-before and not morally strong. A program containing a data race has
undefined behavior.

\paragraph{Safe Rust.}
Safe tensor accesses supply two properties the definition needs. First,
\emph{alias-XOR-mutability across tile programs}: A tile program writes tensor
memory only through its partition view, and the borrow checker keeps that view's sub-tensors
disjoint from every other program's accessed memory. This rests on the
partition map being a function from output sub-tensors to tile programs
(\S\ref{sec:design:escape}): Each sub-tensor has exactly one owner. The
launch grid is derived from the partition
shape, so each tile program is assigned its own sub-tensor, and mutable
partitions can be constructed in safe code only through the partition APIs,
starting from an owned tensor or an exclusive borrow. Mutable views map distinct
in-bounds coordinates to distinct memory locations. For the mapped case, the
launcher validates the tile-program count against the partition grid, and accesses
go through branded \lstinline|PartitionIndex| values (\S\ref{sec:design:escape}),
which the front end accepts only from the program's own
map~\cite{kiselyov2007capabilities,yanovski2021ghostcell}, so every
store lands in bounds on a sub-tensor the map assigned to that program. An immutable \lstinline|&Tensor|
is read-only and cannot alias any partition, because a partition derives from an
owned tensor or an exclusive borrow (\S\ref{sec:design:h2d}). Second, \emph{ordering within a tile program}: The
compiler threads a token chain through every load and store on a
\lstinline|&mut| view (\S\ref{sec:design:tokens}), so accesses from distinct
operations are related in both program order and token order.

Safe accesses to shared scalar globals are non-weak and use device or
system scope. Their storage cannot alias tensor views in safe code.

\begin{theorem}
\label{thm:drf}
Every execution of a kernel written in safe \sysname is data-race-free under
Tile~IR's memory model.
\end{theorem}
\begin{proof}
Consider a conflicting pair $a,b$. At least one, say $a$, is a write.
If $a$ accesses a global, so does $b$. Each access specifies a scope that
contains the other's tile block, so the pair is morally strong.
Otherwise, $a$ is a tensor write. In safe code, tensor stores exist only
on partition views, so $a$ writes through the partition
view of the tile program that issued it. Element writes within one tile store
cannot conflict, so $a$ and $b$ come from distinct tile operations.
If $b$ is issued by a different tile
program, alias-XOR-mutability makes $a$'s partition view disjoint from every
location $b$ can access, so $a$ and $b$ cannot share a location, contradicting
conflict. If $b$ is issued by the same tile program through an immutable view,
the same contradiction follows, because an immutable view aliases no partition.
In the only remaining case, $a$ and $b$ both access the issuing program's
\lstinline|&mut| view, so they are related in program order and token order,
hence in restricted program order, and are therefore morally strong. Every
conflicting pair is thus morally strong, so no execution contains a data race.
\end{proof}

\noindent For tensor accesses, this is the argument that makes ordinary Rust
data-race-free~\cite{jung2018rustbelt}:
Alias-XOR-mutability removes conflicts between threads and ordering removes
them within a thread, with Tile~IR's token order (``waits-for'') in the role
that sequenced-before plays in the C++ and Rust memory models. The store-index
bug of \S\ref{sec:design:race} is the canonical execution it excludes: A store
through the plain partition view has no programmer-chosen destination index, so the conflicting
write to another program's sub-tensor is unconstructible.

\paragraph{Host execution.}
Across launches, the host runtime is trusted to preserve operand storage
and access rights until device work completes, including when operations
are cancelled or forgotten. Conflicting tensor accesses are excluded or
ordered by completion, stream order, or graph dependencies. Global accesses
remain morally strong at device or system scope. Under this runtime
contract, data-race freedom extends to safe host programs.

\section{Bounds-Check Census}
\label{app:census}

\ifappCensusGemm
\else
Table~\ref{tab:census_attention} lists every bounds-check obligation of
the attention kernel of \S\ref{sec:placement} and where the compiler
resolves it. GEMM's four obligations are walked in \S\ref{sec:placement}.
\fi GEMM's dispositions were read from the
compiled launcher's check list and confirmed by launching each violating
shape, which the host rejects with the predicate's own message while
the kernel's Tile~IR contains no assert. Attention's dispositions were
read from the Tile~IR dump, in which the four hoisted asserts sit
immediately before the inner loop, and from the launcher's check list.
The attention census is the same at every tile shape measured.
The evidence is in the paper's repository under
\texttt{benchmarks/exp4\_hoisting}. In Listing~\ref{lst:attention_excerpt},
the literal third coordinates are within the static head dimension $D$
and are discharged at specialization. The \lstinline|q| coordinates come
from the output's map and become launch preconditions. The four checks on
the \lstinline|k| and \lstinline|v| loads are hoisted to the preheader of
the \lstinline|j| loop: \lstinline|kvh| is loop-invariant there, and
\lstinline|j| is affine with extreme $\mathit{tc}-1$. Both \lstinline|kvh| and \lstinline|tc| are decoded from the partition
index inside the persistent loop, so the \lstinline|j| preheader is the earliest
point at which they exist.

\begin{table}[h]
  \caption{Causal prefill attention (Listing~\ref{lst:attention_excerpt}):
    nine obligations, three discharged at specialization, two launch
    preconditions, four hoisted to the preheader of the \texttt{j} loop.
     $q_i$ abbreviates $\mathit{dim}(q, i)$. $D$ is the static head dimension.}
  \label{tab:census_attention}
  \centering
  \footnotesize
  \setlength{\tabcolsep}{2.5pt}
  \begin{tabular}{llll}
    \toprule
    access & index & where & why \\
    \midrule
    \texttt{q} axis 0 & \texttt{qm}    & launch    & $\lceil \mathit{out}_0/BM \rceil \le \lceil q_0/BM \rceil$ \\
    \texttt{q} axis 1 & \texttt{qh} & launch    & $\lceil \mathit{out}_1/1 \rceil \le \lceil q_1/1 \rceil$ \\
    \texttt{q} axis 2 & \texttt{0}            & discharged & literal within $D$ \\
    \texttt{k},\texttt{v} axis 0 & \texttt{kvh} & preheader & loop-invariant \\
    \texttt{k},\texttt{v} axis 1 & \texttt{j}   & preheader & affine; tests $\mathit{tc}-1$ \\
    \texttt{k},\texttt{v} axis 2 & \texttt{0}   & discharged & literal within $D$ \\
    store             & branded               & none      & proven by the brand \\
    \bottomrule
  \end{tabular}
\end{table}

\begin{figure*}[t]
\centering
\begin{subfigure}[t]{0.49\textwidth}
  \centering
  \input{figures/tikz/exp3_grout_sweep}
  \subcaption{RTX~5090 / Qwen3-4B.}
  \label{fig:grout_sweep_5090}
\end{subfigure}
\hfill
\begin{subfigure}[t]{0.49\textwidth}
  \centering
  \input{figures/tikz/exp3_grout_sweep_b200}
  \subcaption{B200 / Qwen3-32B.}
  \label{fig:grout_sweep_b200}
\end{subfigure}
\caption{Qwen3 single-request performance (f16, batch~1, median with IQR
  shading; prefix caching disabled): Prefill latency vs.\ prompt length
  $p_p$ ($t_g{=}36$) and decode throughput vs.\ generation length $t_g$
  ($p_p{=}18$). }
\label{fig:grout_sweeps}
\end{figure*}

\begin{figure}[h]
\begin{lstlisting}[caption={The obligation-carrying lines of the causal prefill
  attention kernel of \S\ref{sec:placement}: three loads over three axes inside a
  persistent loop over the output map, with the inner \texttt{j} loop
  over key/value tiles.},
  label={lst:attention_excerpt}]
for index in out.iter_indices() {  // persistent loop
  let [qm, qh, _] = index.coords();
  let kvh = qh / grp;              // kv head of qh
  let tq = q_part.load([qm, qh, 0]);
  let tc = ceil_div(min(q0 + (qm + 1) * BM, kv_len), BN);
  for j in 0..tc {          // 4 checks hoist to here
    let tk = k_part.load([kvh, j, 0]);
    // ... QK^T, causal mask, online softmax ...
    let tv = v_part.load([kvh, j, 0]);
    // ... accumulate ...
  }
  out.store(acc, index);
}
\end{lstlisting}
\end{figure}

\section{\appDTitle}
\label{app:exp3_extras}

\begin{table}[t]
  \caption{Single-pipeline completion-path cost (RTX~5090): one launch and
    wait of the element-wise kernel of \S\ref{sec:eval:execution} over $d$
    elements. Columns: sync wall time; async's increase over sync with the
    default $20~\mu$s spin and with the spin disabled, in $\mu$s and as a
    fraction of sync; the share of Qwen3-4B and Qwen3-32B decode launches
    whose kernels fall in the duration band ending at each row's sync
    time.}
  \label{tab:spin_cost}
  \centering
  \footnotesize
  \setlength{\tabcolsep}{4pt}
  \begin{tabular}{lrrrrr}
    \toprule
    & sync & \multicolumn{2}{c}{async, over sync} & \multicolumn{2}{c}{decode launches} \\
    elements $d$ & $\mu$s & spin & no spin & 4B & 32B \\
    \midrule
    $2^{21}$          & 7.8    & ${+}0.0$ (0.3\%) & ${+}1.6$ (21\%)  & 59\% & 62\% \\
    $2^{24}$          & 20.3   & ${+}0.1$ (0.7\%) & ${+}1.9$ (9.2\%) & 8\%  & 23\% \\
    $17 \cdot 2^{20}$ & 43.3   & ${+}1.7$ (3.9\%) & ${+}1.7$ (4.0\%) & 25\% & 8\% \\
    $2^{25}$          & 122.9  & ${+}2.4$ (1.9\%) & ${+}2.4$ (1.9\%) & 8\%  & 8\% \\
    $9 \cdot 2^{23}$  & 291.1  & ${+}2.5$ (0.8\%) & ${+}2.5$ (0.9\%) & 0\%  & 0.1\% \\
    \bottomrule
  \end{tabular}
\end{table}

Table~\ref{tab:spin_cost} isolates what each completion path costs a
single pipeline. The three paths run in rotation for ten rounds of 50
samples at default clocks. A round's value is its lower quartile, which
selects the same one of the workstation's two timing states for all three
paths, and each entry is the median over rounds (for async, of that
round's difference to sync). The launch-share bands come from the decode
profiles of \S\ref{sec:eval:grout}. In the first two rows, the work
remaining at the await fits the spin budget and spin-polled async matches
sync within $0.2~\mu$s: The wait resolves at the poll site. From the
third row on, the spin expires and both async paths pay the reactor's wake
path, about $2~\mu$s over sync, whether or not the spin ran. The expired
spin consumes CPU time and adds no latency of its own. The fraction
column shows the amortization: The wake path is a fifth of an $8~\mu$s
launch and wait, two percent at $123~\mu$s, and about one percent at
$290~\mu$s, while the spin removes it for the short kernels where it
matters. The launch columns place the Qwen3 decode kernels in these
bands: About 60\% of launches are shorter than the first row, where the
spin path costs under $0.1~\mu$s and the reactor path alone would cost a fifth,
and nearly all the rest fall in the next three rows, where the default
spin keeps async within $2.4~\mu$s, under four percent, of sync. A decode
step chains 437 or 838 such kernels with one await at the end
(Figure~\ref{fig:exec_mode_grout}), so the step pays the wake path once
per 6 to 12~ms of GPU work, and the replayed steps agree within the 2\%
measurement spread. For an isolated pipeline the budget is therefore
harmless at any kernel size. \AppOverlap{With many streams, each in-flight operation can occupy the
shared submission thread for up to the budget, which accounts for the
gap between the two async lines in Figure~\ref{fig:bimodal_throughput}. }The
default setting gives individual pipelines the same latency as sync
(Figure~\ref{fig:exec_mode_scaling}).\AppOverlap{ Multi-stream drivers
disable it in favor of reactor-based completion.}

\section{Qwen3 Inference Sweeps}
\label{app:grout_sweeps}

Figure~\ref{fig:grout_sweeps} shows the full prompt-length and
generation-length sweeps behind Table~\ref{tab:grout} of
\S\ref{sec:eval:grout}, with interquartile spreads. All measurements are single-request,
batch~1. On the B200, each cell is the median of 10 repetitions. On the
RTX~5090, each cell is the median of 10 repetitions up to
$p_p, t_g{=}512$ and 3 repetitions beyond. Prefill latency is each
engine's per-request timing. vLLM's offline API does not expose it, so
we use time-to-first-token from a \texttt{max\_tokens=1} probe. vLLM
0.18.0 and SGLang 0.5.9 run the same f16 weights with CUDA graphs enabled
and prefix or radix caching disabled. Every cell follows three warmup
requests and uses greedy decoding with the end-of-sequence token ignored
to fix the generation length. At
$t_g{=}8192$ decode reaches 75\%, 73\%, and 74\% of the memory bandwidth roofline of
\S\ref{sec:eval:grout} on the RTX~5090 and 67\%, 65\%, and 64\% on the
B200 for \enginename, vLLM, and SGLang, and each engine's fraction moves
by at most 3 points across the generation sweep.

\ifappRuntime
\section{Host Runtime Details}
\label{app:runtime}

This appendix carries the implementation detail behind
\S\ref{sec:execution}: the \lstinline|DeviceOp| trait and a generated
implementation, the combinators, the kernel entry function, and the
completion path.

\begin{figure}[h]
\begin{lstlisting}[caption={The \texttt{DeviceOp} trait and the \texttt{execute}
  generated for the \texttt{add} launcher of Listing~\ref{lst:vec_add}.
  \texttt{execute} is unsafe because the output's device writes may still be in
  flight; safe \texttt{sync} runs it on a stream, synchronizes, and only then
  returns the recovered arguments.},
  label={lst:deviceop}]
trait DeviceOp {
  type Output: Send;
  unsafe fn execute(self, cx: &ExecutionContext)
    -> Result<Self::Output>;
  fn sync(self) -> Result<Self::Output> {
    let cx = ExecutionContext::new(next_stream()?);
    let out = unsafe { self.execute(&cx) };
    cx.stream().synchronize()?;
    out
  }
}

// proc macro builds AddLaunch + DeviceOp impl
impl DeviceOp for AddLaunch {
  type Output = (Partition<Tensor<f32>>, Tensor<f32>, Tensor<f32>);
  unsafe fn execute(self, cx: &ExecutionContext) -> Result<Self::Output> {
    let (z, x, y) = self.args.execute(cx)?;
    let f = self.cached_or_jit(cx)?;
    launch(f, &z, &x, &y, cx.stream());
    Ok((z.recover(), x.recover(), y.recover()))
  }
}
\end{lstlisting}
\end{figure}

The \lstinline|.sync()| call in Listing~\ref{lst:vec_add} runs the launcher's
\lstinline|execute| method, which takes an execution context (the CUDA stream, device,
and memory pool). It first runs the operations that produce the kernel's
arguments: in Listing~\ref{lst:vec_add}, the \lstinline|zeros| and
\lstinline|ones| allocations and the \lstinline|partition| call, which
asynchronously allocate \lstinline|z|, \lstinline|x|, and \lstinline|y| and
partition the output \lstinline|z|, all on the stream. It then fetches the
kernel from the JIT cache or compiles it, sequences the launch on the same
stream so it runs only after those inputs are ready, and recovers the arguments
in their host types. Because the launch is asynchronous, the returned tensors
may name memory whose writes are still in flight, so \lstinline|execute| cannot
be exposed as a safe operation. The safe \lstinline|sync| method restores
Rust's usual boundary: it runs \lstinline|execute| on a stream, synchronizes,
and only then returns the recovered output. This general pattern ensures that
\lstinline|DeviceOp| composition is safe.

Borrowed tensors are not owned values, so async task APIs that require
owned, independently live operations reject them at compile time. Callers use \lstinline|Arc| when a tensor must be shared with such a
task.

Combinators include \lstinline|then| (chain with data dependency),
\lstinline|zip!| (combine
independent ops), \lstinline|shared| (cloneable, execute-once),
\lstinline|map|, \lstinline|first|, \lstinline|boxed|. Chaining with
\lstinline|then| submits the dependent operation to the same stream, so it
observes the prior operation's writes without explicit synchronization. The generated launcher realizes the host column of
Table~\ref{tab:type_mapping} through two traits: \lstinline|prepare()|
turns a host value into launch arguments and relinquishes host access,
and \lstinline|recover()| returns the same type after execution, so the
borrow checker sees the whole launch, execute, and return lifecycle. An
\lstinline|Arc<Tensor<T>>| comes back as an \lstinline|Arc<Tensor<T>>|,
a borrow as the borrow.

The AST is embedded unlowered because of a proc-macro constraint: Rust
macros run before type checking and monomorphization, while the JIT
sees the later launch-time facts needed for specialization. Rank-polymorphic traits hide the per-rank
generated structs and operation impls from user code, so kernels call
operations by their ordinary names while the implementation dispatches to
the rank that \texttt{rustc} inferred.

\paragraph{Kernel entry generation.}
A CUDA kernel launch passes only raw pointers and scalars and has no
notion of a tensor or partition, so for each source entry point the JIT
generates a device entry function that defines the host/device ABI for
\sysname's tensor types. The launcher of \S\ref{sec:execution} marshals
each pointer together with its shape, stride, and partition scalars. The
function reconstructs the matching Tile~IR views on the device, a
partition view for each \lstinline|&mut Tensor| parameter and a shared
tensor view for each \lstinline|&Tensor| parameter, and then calls the
user's kernel body unchanged.

\paragraph{Completion path.}
The obvious completion mechanism, a host-function callback
(\lstinline|cuLaunchHostFunc|) enqueued behind the kernel, costs tens of
microseconds per operation and blocks the stream until the callback
returns. A first design built on it made every awaited kernel measurably
slower than a synchronous one. The shipped path is staged. An awaited
operation first spin-polls \lstinline|cuStreamQuery| under a bounded
budget, $20~\mu$s by default, counted from the await; $\mu$s-scale
waits are too short for a notification round
trip~\cite{barroso2017killer, kalia2019erpc}, and most waits resolve at
the poll site with no context switch. A wait that outlives the budget
registers with a reactor: The stream writes a completion flag with
\lstinline|cuStreamWriteValue32| into mapped pinned host memory, and one
process-wide scanning thread wakes the registered futures, amortizing
notification across all in-flight
pipelines~\cite{lustig2013fine, axboe2019iouring}. The host-function
path remains only as a fallback. Table~\ref{tab:spin_cost} prices the
two paths against synchronous execution\AppOverlap{, and
Figure~\ref{fig:bimodal_throughput} shows the multi-stream case, where the
budget is set to zero so the shared runtime thread is never held by a
spin}.
\fi

%% file: figures/tikz/exp3_grout_sweep.tex
\begin{tikzpicture}
\definecolor{groutC}{HTML}{D55E00}
\definecolor{sglangC}{HTML}{0072B2}
\definecolor{vllmC}{HTML}{009E73}
\pgfplotsset{
  exp3 axis/.style={
    scale only axis,
    width=0.33\linewidth,
    height=0.14\linewidth,
    xmode=log, log basis x=2,
    axis lines*=left,
    tick align=outside,
    tick style={black, line width=0.4pt},
    axis line style={line width=0.5pt},
    ymajorgrids,
    grid style={black!25, line width=0.3pt},
    xlabel style={yshift=2pt},
    every axis plot/.append style={line width=0.65pt, mark size=1.7pt},
  },
  exp3 vllm/.style={vllmC, mark=triangle*},
  exp3 sglang/.style={sglangC, mark=square*, mark size=1.5pt},
  exp3 grout/.style={groutC, mark=*, line width=0.95pt},
}
\begin{groupplot}[
  group style={group size=2 by 1, horizontal sep=40pt},
  exp3 axis,
]
\nextgroupplot[
  ymode=log,
  xmin=13.2546, xmax=11124.9,
  ymin=5.2158, ymax=1000,
  xtick={18,128,512,2048,8192},
  xticklabels={18,128,512,2048,8192},
  ytick={10,30,100,300,1000},
  yticklabels={10,30,100,300,1000},
  ylabel={ms},
  ylabel style={yshift=-6pt},
  xlabel={Prompt tokens $p_p$},
  legend to name=exp3legendrtx,
  legend columns=-1,
  legend style={
    draw=none, fill=none, /tikz/every even column/.append style={column sep=8pt},
  },
  legend image post style={mark size=1.7pt},
]
\addlegendimage{exp3 grout}\addlegendentry{\textbf{\enginenamens}}
\addlegendimage{exp3 sglang}\addlegendentry{SGLang}
\addlegendimage{exp3 vllm}\addlegendentry{vLLM}
\addplot [name path=ppvllmlo, draw=none, forget plot] table [x=pp, y=lo] {figures/tikz_data/exp3_rtx5090_pp_vllm.dat};
\addplot [name path=ppvllmhi, draw=none, forget plot] table [x=pp, y=hi] {figures/tikz_data/exp3_rtx5090_pp_vllm.dat};
\addplot [vllmC, fill opacity=0.15, draw=none, forget plot] fill between [of=ppvllmlo and ppvllmhi];
\addplot [name path=ppsglanglo, draw=none, forget plot] table [x=pp, y=lo] {figures/tikz_data/exp3_rtx5090_pp_sglang.dat};
\addplot [name path=ppsglanghi, draw=none, forget plot] table [x=pp, y=hi] {figures/tikz_data/exp3_rtx5090_pp_sglang.dat};
\addplot [sglangC, fill opacity=0.15, draw=none, forget plot] fill between [of=ppsglanglo and ppsglanghi];
\addplot [name path=ppgroutlo, draw=none, forget plot] table [x=pp, y=lo] {figures/tikz_data/exp3_rtx5090_pp_grout-v2.dat};
\addplot [name path=ppgrouthi, draw=none, forget plot] table [x=pp, y=hi] {figures/tikz_data/exp3_rtx5090_pp_grout-v2.dat};
\addplot [groutC, fill opacity=0.15, draw=none, forget plot] fill between [of=ppgroutlo and ppgrouthi];
\addplot [exp3 vllm, forget plot] table [x=pp, y=med] {figures/tikz_data/exp3_rtx5090_pp_vllm.dat};
\addplot [exp3 sglang, forget plot] table [x=pp, y=med] {figures/tikz_data/exp3_rtx5090_pp_sglang.dat};
\addplot [exp3 grout, forget plot] table [x=pp, y=med] {figures/tikz_data/exp3_rtx5090_pp_grout-v2.dat};
\nextgroupplot[
  xmin=27.4440, xmax=10745.9,
  ymin=144.4205, ymax=177.7126,
  xtick={36,128,512,2048,8192},
  xticklabels={36,128,512,2048,8192},
  ytick={150,160,170},
  ylabel={tokens/s},
  ylabel style={yshift=-4pt},
  xlabel={Generated tokens $t_g$},
]
\addplot [name path=tgvllmlo, draw=none] table [x=tg, y=lo] {figures/tikz_data/exp3_rtx5090_tg_vllm.dat};
\addplot [name path=tgvllmhi, draw=none] table [x=tg, y=hi] {figures/tikz_data/exp3_rtx5090_tg_vllm.dat};
\addplot [vllmC, fill opacity=0.15, draw=none] fill between [of=tgvllmlo and tgvllmhi];
\addplot [name path=tgsglanglo, draw=none] table [x=tg, y=lo] {figures/tikz_data/exp3_rtx5090_tg_sglang.dat};
\addplot [name path=tgsglanghi, draw=none] table [x=tg, y=hi] {figures/tikz_data/exp3_rtx5090_tg_sglang.dat};
\addplot [sglangC, fill opacity=0.15, draw=none] fill between [of=tgsglanglo and tgsglanghi];
\addplot [name path=tggroutlo, draw=none] table [x=tg, y=lo] {figures/tikz_data/exp3_rtx5090_tg_grout-v2.dat};
\addplot [name path=tggrouthi, draw=none] table [x=tg, y=hi] {figures/tikz_data/exp3_rtx5090_tg_grout-v2.dat};
\addplot [groutC, fill opacity=0.15, draw=none] fill between [of=tggroutlo and tggrouthi];
\addplot [exp3 vllm] table [x=tg, y=med] {figures/tikz_data/exp3_rtx5090_tg_vllm.dat};
\addplot [exp3 sglang] table [x=tg, y=med] {figures/tikz_data/exp3_rtx5090_tg_sglang.dat};
\addplot [exp3 grout] table [x=tg, y=med] {figures/tikz_data/exp3_rtx5090_tg_grout-v2.dat};
\end{groupplot}
\node [anchor=south, inner sep=0pt]
  at ($(group c1r1.north east)!0.5!(group c2r1.north west) + (0,6pt)$)
  {\pgfplotslegendfromname{exp3legendrtx}};
\end{tikzpicture}

%% file: figures/tikz/exp3_grout_sweep_b200.tex
\begin{tikzpicture}
\definecolor{groutC}{HTML}{D55E00}
\definecolor{sglangC}{HTML}{0072B2}
\definecolor{vllmC}{HTML}{009E73}
\pgfplotsset{
  exp3 axis/.style={
    scale only axis,
    width=0.33\linewidth,
    height=0.14\linewidth,
    xmode=log, log basis x=2,
    axis lines*=left,
    tick align=outside,
    tick style={black, line width=0.4pt},
    axis line style={line width=0.5pt},
    ymajorgrids,
    grid style={black!25, line width=0.3pt},
    xlabel style={yshift=2pt},
    every axis plot/.append style={line width=0.65pt, mark size=1.7pt},
  },
  exp3 vllm/.style={vllmC, mark=triangle*},
  exp3 sglang/.style={sglangC, mark=square*, mark size=1.5pt},
  exp3 grout/.style={groutC, mark=*, line width=0.95pt},
}
\begin{groupplot}[
  group style={group size=2 by 1, horizontal sep=40pt},
  exp3 axis,
]
\nextgroupplot[
  ymode=log,
  xmin=13.2546, xmax=11124.9,
  ymin=10, ymax=1000,
  xtick={18,128,512,2048,8192},
  xticklabels={18,128,512,2048,8192},
  ytick={10,30,100,300,1000},
  yticklabels={10,30,100,300,1000},
  ylabel={ms},
  ylabel style={yshift=-6pt},
  xlabel={Prompt tokens $p_p$},
  legend to name=exp3legendb200,
  legend columns=-1,
  legend style={
    draw=none, fill=none, /tikz/every even column/.append style={column sep=8pt},
  },
  legend image post style={mark size=1.7pt},
]
\addlegendimage{exp3 grout}\addlegendentry{\textbf{\enginenamens}}
\addlegendimage{exp3 sglang}\addlegendentry{SGLang}
\addlegendimage{exp3 vllm}\addlegendentry{vLLM}
\addplot [name path=ppvllmlo, draw=none, forget plot] table [x=pp, y=lo] {figures/tikz_data/exp3_b200_pp_vllm.dat};
\addplot [name path=ppvllmhi, draw=none, forget plot] table [x=pp, y=hi] {figures/tikz_data/exp3_b200_pp_vllm.dat};
\addplot [vllmC, fill opacity=0.15, draw=none, forget plot] fill between [of=ppvllmlo and ppvllmhi];
\addplot [name path=ppsglanglo, draw=none, forget plot] table [x=pp, y=lo] {figures/tikz_data/exp3_b200_pp_sglang.dat};
\addplot [name path=ppsglanghi, draw=none, forget plot] table [x=pp, y=hi] {figures/tikz_data/exp3_b200_pp_sglang.dat};
\addplot [sglangC, fill opacity=0.15, draw=none, forget plot] fill between [of=ppsglanglo and ppsglanghi];
\addplot [name path=ppgroutlo, draw=none, forget plot] table [x=pp, y=lo] {figures/tikz_data/exp3_b200_pp_grout-v2.dat};
\addplot [name path=ppgrouthi, draw=none, forget plot] table [x=pp, y=hi] {figures/tikz_data/exp3_b200_pp_grout-v2.dat};
\addplot [groutC, fill opacity=0.15, draw=none, forget plot] fill between [of=ppgroutlo and ppgrouthi];
\addplot [exp3 vllm, forget plot] table [x=pp, y=med] {figures/tikz_data/exp3_b200_pp_vllm.dat};
\addplot [exp3 sglang, forget plot] table [x=pp, y=med] {figures/tikz_data/exp3_b200_pp_sglang.dat};
\addplot [exp3 grout, forget plot] table [x=pp, y=med] {figures/tikz_data/exp3_b200_pp_grout-v2.dat};
\nextgroupplot[
  xmin=27.4440, xmax=10745.9,
  ymin=72.2244, ymax=84.8231,
  xtick={36,128,512,2048,8192},
  xticklabels={36,128,512,2048,8192},
  ytick={75,80},
  ylabel={tokens/s},
  ylabel style={yshift=-4pt},
  xlabel={Generated tokens $t_g$},
]
\addplot [name path=tgvllmlo, draw=none] table [x=tg, y=lo] {figures/tikz_data/exp3_b200_tg_vllm.dat};
\addplot [name path=tgvllmhi, draw=none] table [x=tg, y=hi] {figures/tikz_data/exp3_b200_tg_vllm.dat};
\addplot [vllmC, fill opacity=0.15, draw=none] fill between [of=tgvllmlo and tgvllmhi];
\addplot [name path=tgsglanglo, draw=none] table [x=tg, y=lo] {figures/tikz_data/exp3_b200_tg_sglang.dat};
\addplot [name path=tgsglanghi, draw=none] table [x=tg, y=hi] {figures/tikz_data/exp3_b200_tg_sglang.dat};
\addplot [sglangC, fill opacity=0.15, draw=none] fill between [of=tgsglanglo and tgsglanghi];
\addplot [name path=tggroutlo, draw=none] table [x=tg, y=lo] {figures/tikz_data/exp3_b200_tg_grout-v2.dat};
\addplot [name path=tggrouthi, draw=none] table [x=tg, y=hi] {figures/tikz_data/exp3_b200_tg_grout-v2.dat};
\addplot [groutC, fill opacity=0.15, draw=none] fill between [of=tggroutlo and tggrouthi];
\addplot [exp3 vllm] table [x=tg, y=med] {figures/tikz_data/exp3_b200_tg_vllm.dat};
\addplot [exp3 sglang] table [x=tg, y=med] {figures/tikz_data/exp3_b200_tg_sglang.dat};
\addplot [exp3 grout] table [x=tg, y=med] {figures/tikz_data/exp3_b200_tg_grout-v2.dat};
\end{groupplot}
\node [anchor=south, inner sep=0pt]
  at ($(group c1r1.north east)!0.5!(group c2r1.north west) + (0,6pt)$)
  {\pgfplotslegendfromname{exp3legendb200}};
\end{tikzpicture}

%% file: figures/tikz/exp2_async_throughput.tex
\begin{tikzpicture}[baseline=(current axis.north)]
\definecolor{expSyncCh}{HTML}{E69F00}
\definecolor{expAsync}{HTML}{009E73}
\begin{axis}[
  width=\linewidth,
  height=0.69\linewidth,
  xmode=log, ymode=log,
  xlabel={Host work $W$ ($\mu$s)},
  ylabel={Throughput (iters/s)},
  axis lines*=left,
  axis line style={black},
  tick style={black, thin},
  grid=both,
  major grid style={black!12, line width=0.3pt},
  minor grid style={black!8, line width=0.3pt},
  legend style={
    draw=none, fill=none, at={(0.5,-0.66)}, anchor=north, legend columns=2,
    column sep=0.5em,
  },
  legend cell align=left,
  legend image post style={mark size=2pt},
  ymin=85, ymax=4300,
  xmin=0.75, xmax=13000,
]
\addplot[expAsync, fill=expAsync, fill opacity=0.15, draw=none, forget plot]
  table[x=x, y=y] {figures/tikz_data/exp2b_async_band.dat} -- cycle;
\addplot[expSyncCh, fill=expSyncCh, fill opacity=0.15, draw=none, forget plot]
  table[x=x, y=y] {figures/tikz_data/exp2b_sync_band.dat} -- cycle;

\addplot[expAsync, line width=1.1pt, mark=triangle*, mark size=2.2pt]
  table[x=w, y=line] {figures/tikz_data/exp2b_async.dat};
\addlegendentry{async}
\addplot[expSyncCh, line width=1.1pt, mark=square*, mark size=1.8pt]
  table[x=w, y=line] {figures/tikz_data/exp2b_sync.dat};
\addlegendentry{sync}
\end{axis}
\end{tikzpicture}

%% file: figures/tikz/exp3_bimodal_throughput.tex
\begin{tikzpicture}[baseline=(current axis.north)]
\definecolor{exp3serial}{HTML}{D55E00}
\definecolor{exp3threaded}{HTML}{E69F00}
\definecolor{exp3async}{HTML}{009E73}
\begin{axis}[
  width=\linewidth,
  height=0.69\linewidth,
  axis lines*=left,
  xmode=log,
  log basis x=2,
  log ticks with fixed point,
  xtick={1,2,4,8,16,32},
  xminorticks=false,
  xlabel={Streams $S$},
  ylabel={Throughput (WU/ms)},
  ytick={20,40,60,80},
  ymajorgrids,
  grid style={gray!35, line width=0.3pt},
  tick style={black, line width=0.4pt},
  enlarge x limits=0.05,
  enlarge y limits=0.05,
  legend style={
    draw=none, fill=none, at={(0.5,-0.66)}, anchor=north, legend columns=2,
    /tikz/every even column/.append style={column sep=4pt},
  },
  legend cell align=left,
  clip mode=individual,
]

\addplot[exp3serial!70!white, densely dashed, line width=0.8pt]
  table[x=s, y=med] {figures/tikz_data/exp3bt_serial.dat};
\addlegendentry{serial}

\addplot[name path=thrhi, draw=none, forget plot]
  table[x=s, y=hi] {figures/tikz_data/exp3bt_threaded.dat};
\addplot[name path=thrlo, draw=none, forget plot]
  table[x=s, y=lo] {figures/tikz_data/exp3bt_threaded.dat};
\addplot[exp3threaded, fill opacity=0.15, forget plot]
  fill between[of=thrhi and thrlo];
\addplot[exp3threaded, line width=1.0pt, mark=square*, mark size=1.7pt]
  table[x=s, y=med] {figures/tikz_data/exp3bt_threaded.dat};
\addlegendentry{threaded}

\addplot[name path=asyhi, draw=none, forget plot]
  table[x=s, y=hi] {figures/tikz_data/exp3bt_async.dat};
\addplot[name path=asylo, draw=none, forget plot]
  table[x=s, y=lo] {figures/tikz_data/exp3bt_async.dat};
\addplot[exp3async, fill opacity=0.15, forget plot]
  fill between[of=asyhi and asylo];
\addplot[exp3async, line width=1.0pt, mark=triangle*, mark size=2.0pt]
  table[x=s, y=med] {figures/tikz_data/exp3bt_async.dat};
\addlegendentry{async}

\addplot[exp3async, line width=0.9pt, dash pattern=on 2.4pt off 1.4pt,
  mark=triangle, mark size=2.0pt, mark options={solid}]
  table[x=s, y=med] {figures/tikz_data/exp3bt_async_spin.dat};
\addlegendentry{async (spin)}

\end{axis}
\end{tikzpicture}